%% file: main-arXiv.tex
\documentclass[10pt,dvipsnames, onecolumn]{article}
\usepackage[a4paper,hmargin={2.5cm,2.5cm},vmargin={1.0cm,2.5cm},includehead]{geometry}

\input{packages_and_definitions-arXiv.tex}

\input{macros-arXiv.tex}

\begin{document}

\title{Blockchain Security when Messages are Lost}
\author{Taha Ameen, Suryanarayana Sankagiri, Bruce Hajek}
\affil{ \vspace{-0.1 in} University of Illinois at Urbana-Champaign}
\date{}

\markboth{Blockchain Security when Messages are Lost}{}

\maketitle

\begin{abstract}
    \input{00_Abstract.tex}    
\end{abstract}

\section{Introduction}
    \input{01_Introduction.tex}

\section{System Model}
    \input{02_System_Model.tex}

\section{Main Result}
    \input{03_Main_Result.tex}

\section{Definitions and Preliminaries}
    \input{04_Definitions.tex}

\section{Proof Outline} \label{S: 05-Results}
    \input{05_Results.tex}

\section{Conclusion} 
    \input{06_Conclusion.tex}    

\bibliographystyle{alpha}
\bibliography{bibliography.bib}

\appendix
    \input{07_Appendix}



\end{document}

%% file: packages_and_definitions-arXiv.tex
\usepackage{amsmath}
\usepackage{amssymb}
\usepackage{amsfonts}
\usepackage{amsbsy}
\usepackage{amsthm}
\usepackage{mathtools}
\usepackage{bbm}
\usepackage{xfrac}
\usepackage{authblk}

\usepackage[utf8]{inputenc}
\usepackage[english]{babel}
\usepackage[T1]{fontenc}
\usepackage{stmaryrd}
\usepackage{csquotes}
\usepackage{fancyhdr}
\usepackage{enumitem}

\usepackage{subfigure}
\usepackage{float}
\usepackage{graphicx}
\usepackage{booktabs}
\usepackage{multirow, multicol}
\usepackage{color}
\usepackage{url}
\usepackage{hyperref}
\usepackage{cite}
\usepackage{pdfpages}

\usepackage[framed,numbered]{matlab-prettifier}
\usepackage[linesnumbered,ruled]{algorithm2e}

\newtheorem{theorem}{Theorem}[section]
\newtheorem{lemma}[theorem]{Lemma}

\newtheorem{remark}[theorem]{Remark}
\newtheorem{definition}[theorem]{Definition}

\makeatletter
\newtheorem*{rep@theorem}{\rep@title}
\newcommand{\newreptheorem}[2]{%
\newenvironment{rep#1}[1]{%
 \def\rep@title{#2 \ref{##1}}%
 \begin{rep@theorem}}%
 {\end{rep@theorem}}}
\makeatother

\newreptheorem{theorem}{Theorem}
\newreptheorem{lemma}{Lemma}
\newreptheorem{claim}{Claim}
\newreptheorem{definition}{Definition}


%% file: macros-arXiv.tex
\renewcommand{\P}[1]{\mathbb{P}\left(#1 \right)} 
\newcommand{\E}[1]{\mathbb{E}\left[ #1 \right]} 
\newcommand{\N}{\mathbb{N}} 
\renewcommand{\S}{\mathbb{S}} 

\definecolor{warningcol}{rgb}{.99,.1,.5}
\definecolor{todocol}{rgb}{.4,.4,.8}
\definecolor{sketchcol}{rgb}{.4,.4,.8}
\definecolor{outlinecol}{rgb}{.8,.4,.3}

\newcommand{\eps}{\varepsilon} 

\newcommand{\aln}[1]{\begin{align}  #1 \end{align}}
\newcommand{\alns}[1]{\begin{align*}  #1 \end{align*}}
\newcommand{\alnd}[1]{\begin{aligned}  #1 \end{aligned}}

\newcommand{\cas}[1]{\begin{cases}  #1 \end{cases}}
\newcommand{\pf}[1]{\begin{proof}  #1 \end{proof}}

\newcommand{\pbr}[1]{\left( #1 \right)} 
\newcommand{\sbr}[1]{\left[ #1\right]}
\newcommand{\cbr}[1]{\left\{ #1\right\}}

\newcommand{\height}[2]{\mathsf{height}_{#1}\pbr{#2}}

\newcommand{\delay}[2]{\mathsf{delay}\pbr{ b_{#1} \to b_{#2}}}

\newcommand{\MB}{\mathsf{MB}} 
\newcommand{\MBt}{\mathsf{MB}(t)} 
\newcommand{\FRSH}{\mathsf{FS}} 
\newcommand{\BRSH}{\mathsf{BS}}

\newcommand{\FRU}{\mathsf{FU}} 
\newcommand{\BRU}{\mathsf{BU}}

\newcommand{\A}[2]{a(#1,#2)}

\newcommand{\XDAG}{\mathsf{TG}} 

\newcommand{\sep}{\gamma}
\newcommand{\Ev}{\mathsf{E}}

\newcommand{\EA}{\mathsf{E}_{\mathsf{A}}} 
\newcommand{\EB}{\mathsf{E}_{\mathsf{B}}}

\renewcommand{\c}{\mathsf{c}}

\newcommand{\D}{\mathsf{D}} 
\newcommand{\Dc}{\mathsf{D}^{\c}}

\newcommand{\G}{\mathsf{G}} 
\newcommand{\Gc}{\mathsf{G}^{\c}} 

\newcommand{\sfH}{\mathsf{H}} 
\newcommand{\Hc}{\mathsf{H}^{\c}} 

\newcommand{\T}{\mathsf{T}} 

\newcommand{\W}{\mathsf{W}} 
\newcommand{\X}{\mathsf{X}} 
\newcommand{\Y}{\mathsf{Y}} 
\newcommand{\sfZ}{\mathsf{Z}}

\newcommand{\FS}[2]{ \FRSH_{#1}\pbr{ b_{#2}} }
\newcommand{\BS}[2]{ \BRSH_{#1}\pbr{ b_{#2}} }

\newcommand{\FU}[2]{ \FRU_{#1}\pbr{b_{#2}}}
\newcommand{\BU}[2]{ \BRU_{#1}\pbr{b_{#2}}}

\newcommand{\Njss}{ \mathsf{N}_{j}^{\pbr{\ss}} }

\newcommand{\Un}[3]{\mathsf{Unheard}_{#1}\pbr{b_{#2}^{#3}}}
\mathchardef\mhyphen="2D
\newcommand{\dash}{\mhyphen}

\makeatletter
\DeclareRobustCommand{\cev}[1]{%
  \mathpalette\do@cev{#1}%
}
\newcommand{\do@cev}[2]{%
  \fix@cev{#1}{+}%
  \reflectbox{$\m@th#1\vec{\reflectbox{$\fix@cev{#1}{-}\m@th#1#2\fix@cev{#1}{+}$}}$}%
  \fix@cev{#1}{-}%
}
\newcommand{\fix@cev}[2]{%
  \ifx#1\displaystyle
    \mkern#23mu
  \else
    \ifx#1\textstyle
      \mkern#23mu
    \else
      \ifx#1\scriptstyle
        \mkern#22mu
      \else
        \mkern#22mu
      \fi
    \fi
  \fi
}

\makeatother

\newcommand{\Bf}[3]{\overrightarrow{\mathsf{B}}_{#1,#2}^{\pbr{#3}}} 
\newcommand{\Bb}[3]{\overleftarrow{\mathsf{B}}_{#1,#2}^{\pbr{#3}}}
\newcommand{\Bfc}[3]{ \sbr{\overrightarrow{\mathsf{B}}_{#1,#2} ^{\pbr{#3}}} ^{\mathsf{c}}} 
\newcommand{\Bbc}[3]{ \sbr{\overleftarrow{\mathsf{B}}_{#1,#2} ^{\pbr{#3}}} ^{\mathsf{c}}} 
\newcommand{\Mt}{M(t)} 

\newcommand{\ceil}[1]{\left\lceil #1 \right\rceil} 

\newcommand{\Chain}[2]{\mathcal{C}_{#1}(#2)}

\renewcommand{\ss}{\eta} 

\newcommand{\Nb}[2]{\overleftarrow{\mathsf{E}}_{#1}^{\pbr{#2}} } 
\newcommand{\Nf}[2]{\overrightarrow{\mathsf{E}}_{#1}^{\pbr{#2}} }

%% file: 00_Abstract.tex
Security analyses for consensus protocols in blockchain research have primarily focused on the synchronous model, where point-to-point communication delays are upper bounded by a known finite constant. These models are unrealistic in noisy settings, where messages may be lost (i.e. incur infinite delay). In this work, we study the impact of message losses on the security of the proof-of-work longest-chain protocol. We introduce a new communication model to capture the impact of message loss called the \( 0\dash\infty\) model, and derive a region of tolerable adversarial power under which the consensus protocol is secure. The guarantees are derived as a simple bound for the probability that a transaction violates desired security properties. Specifically, we show that this violation probability decays almost exponentially in the security parameter. Our approach involves constructing combinatorial objects from blocktrees, and identifying random variables associated with them that are amenable to analysis. This approach improves existing bounds and extends the known regime for tolerable adversarial threshold in settings where messages may be lost.

%% file: 01_Introduction.tex
Blockchain is the data structure used by peers (miners) in a peer-to-peer network to maintain a common ledger in a decentralized manner. The consistency of this ledger is ensured through consensus protocols such as the longest-chain protocol. Following this protocol, an honest miner groups transactions into a block and appends its block to the longest chain in its view, before broadcasting the new blockchain to all other peers. Further, the system may have adversarial users that deviate from the protocol arbitrarily. Despite adversarial users attempting to disrupt the system and peer-to-peer communication incurring message delays, it is desirable that the parties following the protocol agree on a consistent ledger. 

Blockchain security has been studied under various consensus protocols (see~\cite{Bano_2019, Garay_2020} for a survey). Of these, the longest-chain protocol is of great interest, due its heavy use in modern blockchain implementations. The longest-chain protocol has been modeled under various assumptions: for example, discrete time is used in~\cite{Gazi_2020, Blum_2020}, and continuous time dynamics is used in~\cite{Li_2021, Ren_2019, DKT_2020}. Further, the protocol has also been studied for a variety of leader election mechanisms in the consensus protocol. For instance,~\cite{Pass_2017B, Ren_2019, Garay_2020B} assume the proof-of-work mechanism, whereas~\cite{Pass_2017, Kiayias_2017, Fan_2017} assume a proof-of-stake mechanism. All these works establish security of the longest-chain protocol for the synchronous communication model, where communication delays are upper bounded by a known finite constant. A common theme among these results is that in the synchronous delay model, the longest-chain protocol is `secure' under sufficient honest representation, with high probability. 

In this work, we analyze the impact of message losses on the security of the longest-chain protocol following proof-of-work leader election, by introducing and analyzing an appropriate communication network model. We motivate this by reviewing some existing communication models in the literature and the known security guarantees associated with them.

\subsection{Related Work}

The underlying communication network can delay the successful delivery of peer-to-peer message broadcasts. Popular blockchains such as Bitcoin use the Internet as their communication network. Since this communication is subject to delay, it is natural to model the delays incurred by each block, and study the impact of delay on the security of the longest-chain protocol.

Let \(0 \leq i < j\). Let \(b_i\) represent the \(i\)-th mined honest block. Let \( \delay{i}{j}\) denote the time taken for block \(b_i\) to reach the miner of block \(b_j\), and let \(\beta\) represent the fraction of adversarial computational power in the system. Finally, let \(\lambda\) be the rate at which blocks are mined in the system. Various descriptions of \(\delay{i}{j}\) lead to different communication network models:

\paragraph{Instantaneous Model} The original white-paper by Satoshi Nakamoto \cite{Nakamoto_2008} assumes an ideal communication channel, i.e. \( \delay{i}{j} = 0 \). In this model, the longest-chain protocol is provably secure when the honest computational power in the system exceeds the adversarial computational power, i.e. when \( \beta < 1 - \beta\), or equivalently, when \( \beta < 1/2\).

\paragraph{Synchronous Model} The model assumes a deterministic delay for each block that is upper bounded by a known constant \(\Delta\), i.e., \( \delay{i}{j} \leq \Delta < \infty \). This delay effectively reduces the growth rate of the chain held by an honest user. Even so, it has been proved~\cite{DKT_2020, Gazi_2020} that the synchronous model is secure with high probability if and only if
\alns{ 
\beta < \frac{1}{1+\pbr{1-\beta}\lambda \Delta} \pbr{1 - \beta},
}
where \(\lambda\) is the total mining rate of the honest users.

\paragraph{Partially Synchronous Model}
The partially synchronous model assumes the existence of some unknown and adversarially chosen `Global Stabilization Time (\(\mathsf{GST}\))' such that the delays are unbounded before \(\mathsf{GST}\), but bounded after it~\cite{Dwork_1988}. Therefore, at any time \(t\), the delay satisfies \(\delay{i}{j} \leq \Delta + \max\pbr{0, \mathsf{GST}-t}\). If certain conditions are met, the partially synchronous model is known to be secure with high probability after the Global Stabilization Time~\cite{Neu_2021}.

\paragraph{Sleepy Model}
The sleepy model considers the setting where miners may either be online or offline, and their participation status may change during the execution of the protocol~\cite{Pass_2017}. Let \(h_i\) denote the miner of block \(b_i\). The incurred delay is thus
\alns{ 
\delay{i}{j} = \cas{ 0 & \text{\(h_j\) is awake when \(b_i\) is mined} \\ \infty & \text{\(h_j\) is asleep when \(b_i\) is mined}}.
}
Pass and Shi~\cite{Pass_2017} showed that consensus can be achieved in the sleepy model with high probability, if a majority of the awake miners at any point in time are honest.

\paragraph{Random Delay Model}
The random delay model assumes that the point-to-point delays are independent and identically distributed, i.e. \(\delay{i}{j} \sim \mathsf{X}\), where \( \mathsf{X}\) is some known distribution.  The longest-chain protocol is shown to be secure with high probability in the random delay model, if the delay distribution satisfies certain conditions and the adversarial representation in the system is below a certain threshold~\cite{Sankagiri_2020}. 

Except for the random delay model, none of the above models account for the possibility that point-to-point communication may incur infinite delay, i.e. messages may be lost at random. For instance, the sleepy model allows infinite delay for users that are offline, but does not account for noise in the communication process. In contrast, we introduce and analyze a new communication model to study the impact of lost messages on blockchain security.

\subsection{Contributions}

\paragraph{\(0 \dash \infty\) Model} We introduce the \(0\dash\infty\) model, where the delays are independent and identically distributed over the set \(\{0,\infty\}\). Specifically, for any \(i\), \(j \geq 0\) such that \(i < j\):
\alns{ 
\delay{i}{j} = \cas{ 0 & \text{with probability } 1-d \\ \infty & \text{with probability } d}.
}
This simple model postulates that a message sent point-to-point is either immediately received or permanently lost. This delay is independent for each user, and for each block. The modeling choice aligns with our objective of studying the effect of message losses.

We remark that the \(0 \dash \infty\) model is a special case of the i.i.d. random delay model introduced in~\cite{Sankagiri_2020}, which identifies a region of tolerable adversarial power as a function of the delay distribution. Specifically, if \(d\) is the probability of message loss and \(\beta\) is the fraction of computational power in the system that is adversarial, it is shown that the \(0\dash\infty\) model is secure with high probability when \( \beta < \frac{1-2d}{2\pbr{1-d}}\). However, this characterization is not tight for the \(0 \dash \infty\) model, and the analysis in~\cite{Sankagiri_2020} breaks down in the high-noise regime. For example, security of the model cannot be established when \(d > 1/2\), i.e. more than half the messages are lost on average.

It is reasonable to wonder if adversarial computational power can at all be tolerated in the high-noise regime, for instance, when almost all messages are lost. Our work answers this question in the affirmative, by expanding the known security threshold for the \(0 \dash \infty\) model. In particular, our sufficient condition for security is \( \frac{\beta}{1-\beta} < 1-d\). Figure~\ref{fig: Region} shows this improvement.

\begin{figure}[t]
    \centering
    \includegraphics[width = 0.6 \textwidth]{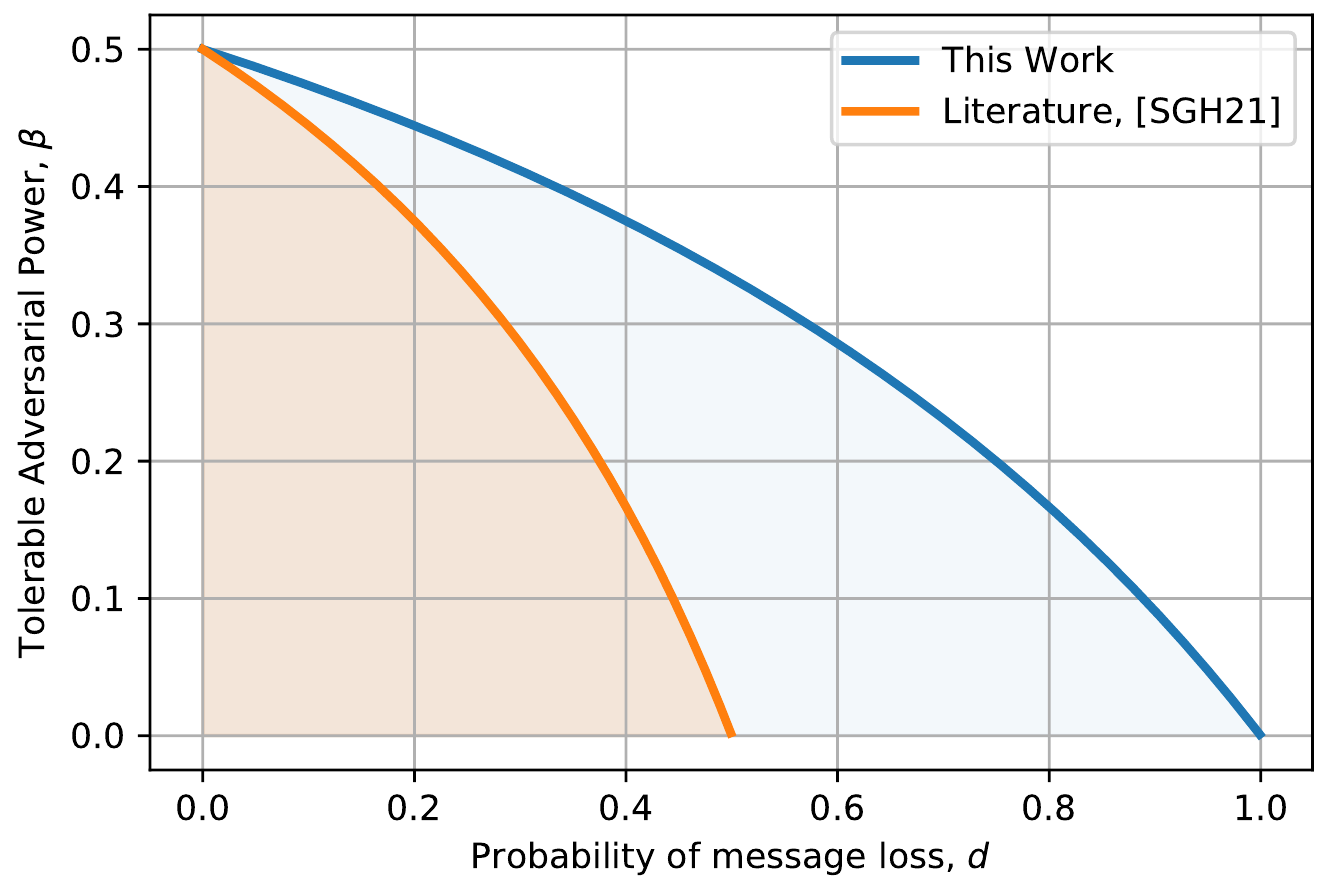}
    \caption{Characterizing the region of tolerable adversarial power}
    \label{fig: Region}
\end{figure}

Our method of analysis is significantly different from that in~\cite{Sankagiri_2020}: we introduce a transmission-graph that captures the history of communication delays between blocks, and identify special paths in the graph that are linked to random variables which are amenable to analysis. Specifically, we identify special objects such as forward-special and backward-special blocks, and associate with them random variables such as forward-unheard and backward-unheard. Our technique also presents a new approach to infer the inclusion of special blocks in the chain held by an honest user through the concept of user-unheard-criterion. The method of analysis is inspired from~\cite{DKT_2020}, where security of the synchronous model is established by considering races between honest and adversarial chains. However, our approach does not rely on message delays being finite, and we hope that the tools we introduce are of utility in the study of consensus mechanisms in more general settings, as well as of independent mathematical interest.
Our contributions are summarized as:
\begin{itemize}
    \item We introduce the \(0\dash\infty\) model as a playground for studying the impact of message losses. This model provides a starting step for more complex models involving message losses.
    \item We introduce combinatorial objects of independent interest such as the transmission-graph. We also identify random variables (forward-unheard, backward-unheard) associated with this graph that are amenable to analysis, and introduce the user-unheard-criterion. These concepts may be utilized in security analysis of blockchain protocols in more general settings.  
    \item We prove that the longest-chain protocol is secure in the \(0\dash\infty\) model if certain conditions are met. These conditions are fairly general, and considerably extend the known threshold of tolerable adversarial power. In this regime, we show that the probability of security violation decays almost exponentially in the security parameter.
\end{itemize}

%% file: 02_System_Model.tex
In this section, we describe our system model. We consider the setting where infinitely many miners participate in the longest-chain protocol for an infinite duration, and use proof-of-work as the leader election mechanism. 

\paragraph{Ledgers, Transactions, Miners, and Blocks} Blockchain is the data structure at the heart of the decentralized mechanism to maintain and update a ledger. The ledger is simply an ordered list of transactions. Transactions are assumed to be available to all the miners as soon as they are made. Miners verify the validity of transactions, and update the ledger by grouping the transactions into blocks and linking blocks to form a blockchain. A block is an abstract data structure that contains a hash pointer to a parent block, a cryptographic signature of the block's miner, transactions, and other metadata. The first block in the system is called the genesis block.

\paragraph{Longest-Chain Protocol and Proof-of-Work} Miners follow the proof-of-work longest-chain protocol for consensus. Following this protocol, a miner groups any and all transactions that are not included in this longest chain into a block, and attempts to append the block to the longest chain in its view. To do so, it must solve a hash puzzle and include the solution as proof-of-work. If the miner is successful, it broadcasts its chain as a message to other miners over a peer-to-peer network, subject to a communication delay. Upon receiving this message, an honest miner adopts the new chain if the received chain is longer than the chain in the miner's memory. Ties are broken using any deterministic rule, for example, by choosing the chain that terminates in the block that hashes to a lower value. The process continues indefinitely. We assume there are infinitely many miners, and at any finite time, a miner who successfully solves the hash puzzle is doing so for the first time almost surely.

\paragraph{Parties in the Protocol} We refer to parties in the protocol as users. Users that contribute to modifying the ledger through appending blocks to the blockchain are called miners. A miner is either honest or corrupt. Honest miners follow the longest-chain protocol, whereas corrupt miners may deviate from the protocol. For simplicity, all corrupt users are grouped into a single entity called the adversary. This allows corrupt miners to communicate instantaneously, and captures the strong setting of perfect coordination between corrupt miners. The adversary can mine on any previously mined block, but its block must contain the proof-of-work to be valid. It can reveal its chain to any subset of honest users, and can delay its message by arbitrary amounts of time. It can also not include all the transactions it knows about that were not in ancestor blocks. We use \(\beta\) to represent the fraction of computational power in the system that is adversarial. 

\paragraph{Mining Process} The mining process is abstracted as follows. Let \(\lambda\) denote the total mining rate of the system. We consider a continuous time model where blocks are mined as a Poisson process with rate \(\lambda\). Since \(\beta\) denotes the fraction of power that is adversarial and since successive mining instances are independent, adversarial block arrivals follow a Poisson process with rate \(\beta \lambda\).

\paragraph{Blockchains and Blocktrees} From any block, a unique sequence of blocks leading back to the genesis block can be identified via the hash pointers. We call this sequence a blockchain, or simply a chain. The convention is that the genesis block is the first block of a chain, and the terminating block is called the tip. At any given slot, honest users store a single chain in their memory. 

\paragraph{Communication Delays} We consider the setting where messages are either instantaneously delivered or permanently lost in an independent and identically distributed manner. Let \(\delay{i}{j}\) denote the delay incurred by block \(b_i\) to reach the miner of block \(b_j\). It is assumed that:
\alns{ 
\delay{i}{j} = \cas{ 0, &\text{ with probability } 1-d \\ \infty, & \text{ with probability } d}.
}
Here, \(d\) is the probability of message loss in an instance of point-to-point communication. 

%% file: 03_Main_Result.tex
In this section, we outline the desired security properties and present our main result. We define security on the level of transactions. It is desirable that a transaction eventually makes it to the ledger, and stays permanently at the same position in the ledger. This notion is formalized in Definition~\ref{def: Security}. 

\begin{definition}[Security] \label{def: Security}
Let \(\tau > 0\). Let \(\mathcal{H}\) be any set of honest users. For any \(h \in \mathcal{H}\), let \(\Chain{h}{t}\) denote the chain held by user \(h\) at time \(t\). We say that a transaction \(\mathsf{tx}\) made at some time \(s\) is \(\pbr{\tau,\mathcal{H}}\)-secure if for any \(h_1,h_2 \in \mathcal{H}\) and any \(s_1,s_2 > s+\tau\), it holds that \(\mathsf{tx}\) is included in a block \(b\) that is at the same position in \(\Chain{h_1}{s_1}\) and \(\Chain{h_2}{s_2}\).
\end{definition}

In the literature, security of a transaction is often defined as the confluence of \emph{persistence} and \emph{liveness}. A transaction satisfies liveness if it is eventually added to the ledger, and it satisfies persistence if it remains in the same position in the ledger for all future time, after a confirmation time. We remark that our definition of security implies these notions of persistence and liveness, and is consistent with existing definitions of security, such as in~\cite{Garay_2015, DKT_2020}. Specifically, if a transaction \(\mathsf{tx}\) satisfies \(\pbr{\tau,\mathcal{H}}\)-security, then it is part of the chain held by all users in \(\mathcal{H}\) before a confirmation time \(\tau\) time has elapsed. Furthermore, once this confirmation time elapses, the transaction remains at the same position in the ledger for all future time. Our main result shows that if certain conditions are satisfied, then any transaction \(\mathsf{tx}\) satisfies \(\pbr{\tau,\mathcal{H}}\) security except with a probability that decays almost exponentially in the confirmation time and scales linearly in the size of \(\mathcal{H}\). It is stated as Theorem~\ref{thm: Main}.

\begin{theorem}[Main Result] \label{thm: Main}
    \input{Theorem-Statements/Thm-Main-Result}
\end{theorem}

The result states that under a certain threshold of tolerable adversarial power, the probability of security violation for any transaction and any finite set of users decays (almost) exponentially in the confirmation time. Hence, this violation probability can be made arbitrarily small by appropriately selecting the  confirmation time. The sufficient condition \( \frac{\beta}{1-\beta} < 1-d\) significantly improves the known threshold of tolerable adversarial power for the \(0\dash\infty\) model (Figure~\ref{fig: Region}). We also remark that our bound for the probability violation comprises of two terms, the latter of which scales linearly in \(|\mathcal{H}|\). This linear scaling is expected, because no single message is successfully transmitted to all users in the model. Therefore, requiring a larger set of users to permanently adopt a transaction in their ledger requires a larger waiting time.  


%% file: Theorem-Statements/Thm-Main-Result.tex
Let \(\beta\) be the fraction of computational power in the system that is adversarial, and \(d\) be the probability of message loss. If 
\(\frac{\beta}{1-\beta} < \pbr{1-d}\), then for every \(\eps > 0\), there exist positive constants \(a\) and \(b\) such that for all \(\tau \geq 0\) and for any honest transaction \(\mathsf{tx}\) and any finite set of honest users \( \mathcal{H}\):
\alns{ 
\P{\mathsf{tx} \text{ violates \( \pbr{\tau, \mathcal{H}}\)-security}} \leq  
\exp\pbr{-a \tau^{1-\eps}} + |\mathcal{H}| \exp \pbr{- b \tau}.
}

%% file: 04_Definitions.tex
This section introduces key quantities that are used extensively in the analysis. In Section~\ref{SS: Graphs}, we introduce the combinatorial objects on which the analysis is performed, such as the main-blocktree and the transmission-graph. Properties of these objects are presented alongside to motivate their purpose. In Section~\ref{SS: Sequences}, key random variables that are amenable to analysis, and associated with special paths in the transmission-graph are identified. These quantities are illustrated through an example in Section~\ref{SS: Example}. In Section~\ref{SS: Nakamoto}, these random variables and are used to define `catch-up events', and the notion of `\(\ss\)-Nakamoto blocks'. Finally, Section~\ref{SS: Observer} introduces the `user-unheard-criterion', which is used to infer useful information about the blockchain held by a user.

\subsection{Graphs and Trees} \label{SS: Graphs}

Three combinatorial objects at the core of our analysis are the main-blocktree and the transmission-graph.

\paragraph{Main-blocktree}
Any block can be uniquely traced back to the first block in the system (called genesis block). The set of all blocks generated (mined) up till time \(t\) forms a directed tree, which we refer to as the main-blocktree and denote it by \(\MB(t)\). Here, \(\MB(t) = (V_t,E_t)\), where the vertex set \(V_t\) is the set of all blocks mined up till time \(t\) and the set of directed edges \(E_t\) comprises all parent-to-child block pairs. \(\MB(t)\) represents the global information about the system, and both honest and adversarial blocks are included in it. Figure~\ref{fig: MB} shows an example of the vicinity of the \(j\)-th honest block, \(b_j\) in some \(\MB(t)\). 

\begin{definition}[Heights]
Let \(b_j\) be the \(j\)-th honest block. The height of \(b_j\) in a blocktree is the length of the directed path (counting edges) from the genesis block to \(b_j\). We denote the height of \(b_j\) in the main-blocktree by \(\height{\MB}{b_j}\).
\end{definition}

\paragraph{Transmission-graph}

At any time \(t\), we associate with the main-blocktree \(\MBt\), a graph consisting of only honest blocks that we call the transmission-graph. The transmission-graph at time \(t\), denoted \(\XDAG(t)\) is a directed acyclic graph that represents the history of network delays among the honest miners. Here, \( \XDAG(t) = \pbr{V_t,E_t}\), where the vertex set \(V_t\) is the set of honestly mined blocks up till time \(t\). An edge between \(b_i\) and \(b_j\) is present if \(\delay{i}{j} = 0\). A useful observation about \(\XDAG(T)\) is presented in Lemma~\ref{lem: XDAG-path-lower-bound}.

\begin{figure}[t]
    \centering
    \subfigure[Main-blocktree]{
    \includegraphics[width = 0.9 \textwidth]{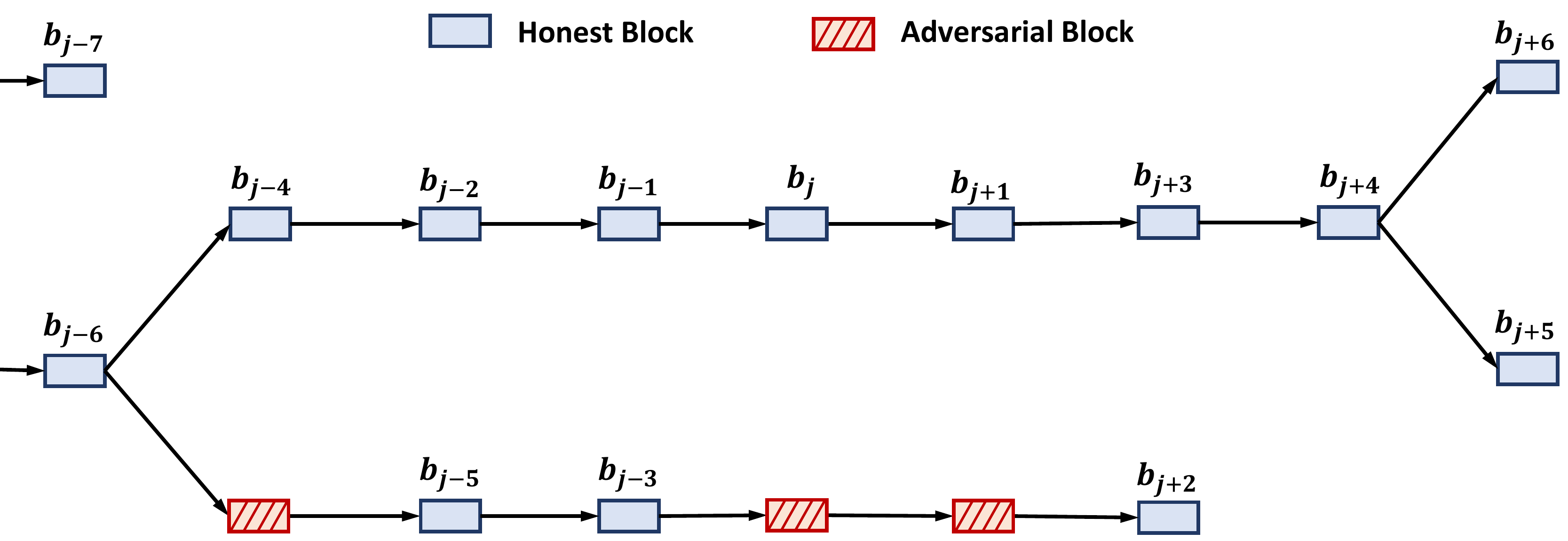}
    \label{fig: MB}
    }
    \subfigure[Transmission-graph]{
    \includegraphics[width = 0.9 \textwidth]{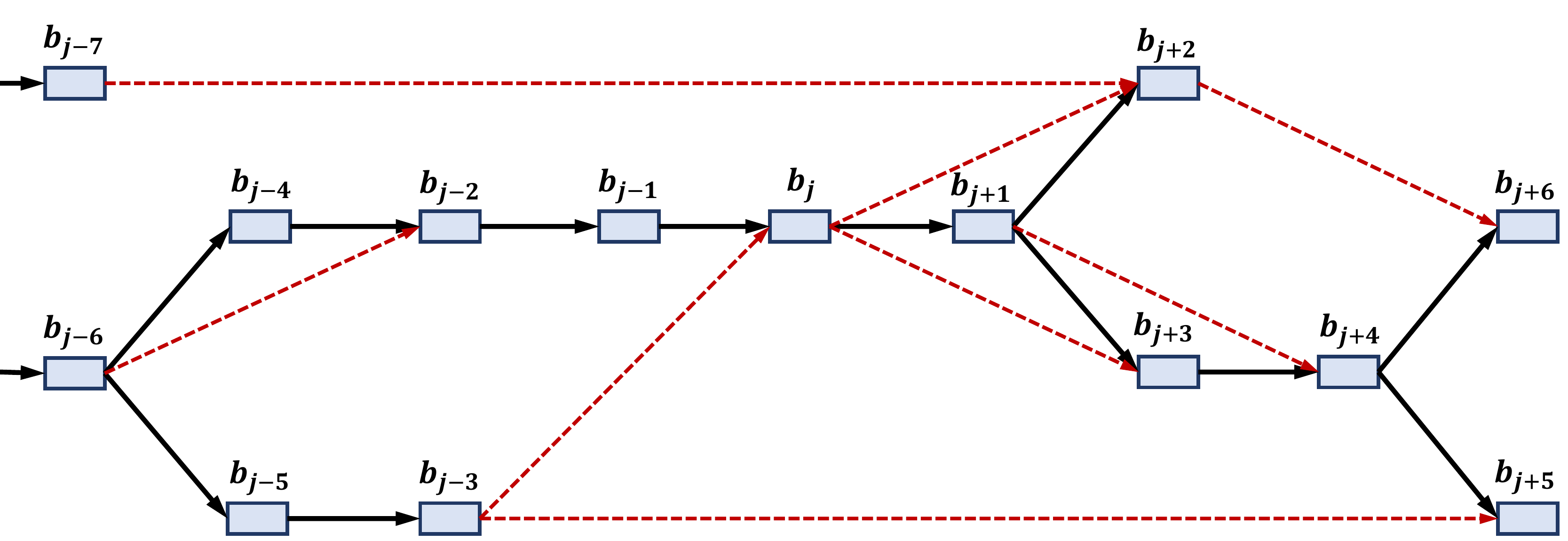}
    \label{fig: XDAG}
    }    
    \caption{Graphs and trees}
\end{figure}

\begin{lemma} \label{lem: XDAG-path-lower-bound}
Let \(b_i\) and \(b_k\) be the \(i\)-th and \(k\)-th honest blocks such that \(i < k\). At any time \(t\), suppose there exists a path \(\mathcal{A}_{i,k}\) of length \(n\) from \(b_i\) to \(b_k\) in \(\XDAG(t)\), i.e.
\[\mathcal{A}_{i,k}\colon b_i = v_0 - v_1 - \cdots - v_n = b_k.\]
Then,
\alns{ 
\height{\MB}{b_k} - \height{\MB}{b_i} &\geq n 
}
\end{lemma}
\pf{ 
If \(v_{i-1}\) and \(v_{i}\) are two blocks in \(\XDAG(t)\) such that there is an edge from \(v_{i-1}\) to \(v_i\), then the miner of block \(v_i\) has heard of block \(v_{i-1}\). Therefore, it adds its block at a height greater than that of \(v_{i-1}\), so we have 
\alns{ 
\height{\MB}{v_{i}} &\geq \height{\MB}{v_{i-1}} + 1
}
Repeatedly applying this inequality over the path from \(b_i\) to \(b_k\) yields the desired result.
}

\subsection{Special Sequences of Honest Blocks} \label{SS: Sequences}

Relative to the \(j\)-th honest block \(b_j\), we define sequences of special blocks that correspond to forward and backward paths in \(\XDAG(t)\). We also define notions of `forward unheard' and `backward unheard'.

\subsubsection{Forward Special Blocks} 
Relative to the \(j\)-th honest block \(b_j\), we define a sequence of `forward special (\(\FRSH\))' blocks as follows.

\begin{definition}[\(j \dash \FRSH\) Sequence]
Let \(j\geq 0\) and let \(b_j\) be the \(j\)-th honest block. The \(j\dash\FRSH\) sequence is a sequence of blocks \((b_{j}^{0},b_{j}^{1},b_{j}^{2},\cdots)\) such that \(b_{j}^{0} = b_j\), and for all \(k \geq 1\), \(b_{j}^{k}\) is the first block to hear of \(b_{j}^{k-1}\).
\end{definition}

\begin{figure}[t]
    \centering
    \includegraphics[width = 0.9 \textwidth]{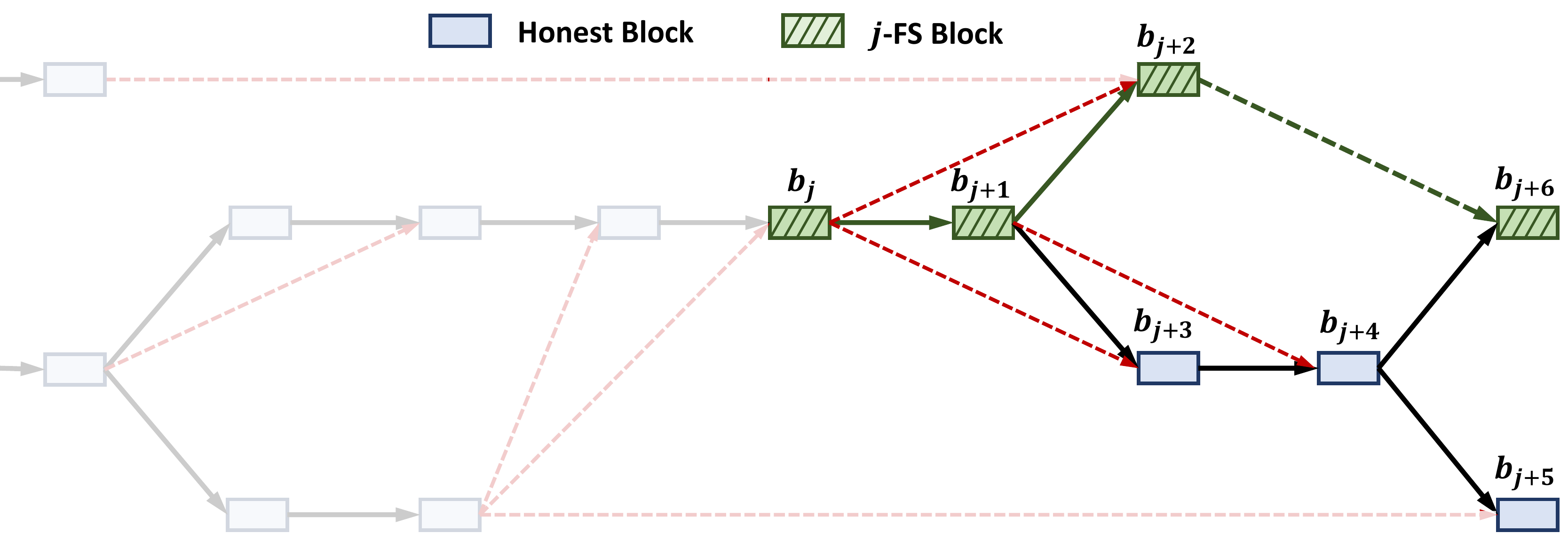}
    \caption{\(j\dash\FRSH\) Sequence}
    \label{fig: XDAG-FRSH}
\end{figure}

We refer to blocks in the \(j\dash\FRSH\) sequence as \(j\dash\FRSH \) blocks. For \(0\leq j < k\), denote by \( \FS{j}{k}\) the number of \( j\dash\FRSH\) blocks mined between \(b_j\) and \(b_{k}\) (inclusive). Note that \( \FS{j}{k-1} \geq 1\), because \(b_j\) is always a \( j\dash\FRSH\) block. The \(j\dash\FRSH\) sequence associated with the transmission-graph in Figure~\ref{fig: XDAG} is shown in Figure~\ref{fig: XDAG-FRSH}. Notice that \( \FS{j}{j+5} = 3\) and \(\FS{j}{j+6} = 4\). This example is explored in more detail in Section~\ref{SS: Example}.

\begin{remark} \label{rem: FS-Dist}
Let \(0 \leq j \leq k\). Let \(d\) be the probability of message loss. The random variable \( \FS{j}{k}\) has the same distribution as \(1 + \sum_{i=1}^{k-j} \mathsf{Be}_i\pbr{1-d}\), where \(\mathsf{Be}_i\pbr{1-d}\) are i.i.d. Bernoulli random variables with success probability equal to \(1-d\). This is because \(b_j\) is a \(j\dash\FRSH\) block, and every subsequent block is independently \(j\dash\FRSH\) with probability \(1-d\).
\end{remark}

\begin{lemma} \label{lem: FRSH-forms-path}
Let \(j \geq 0\). The \(j\dash\FRSH\) sequence \( ( b_{j}^{0} , b_{j}^{1} , b_{j}^{2} , \cdots) \) is a forward directed path in the transmission-graph. Further, if \(k > i \geq 0\), then the heights of the \(j\dash\FRSH\) blocks \(b_{j}^{i}\) and \(b_{j}^{k}\) satisfy:
\alns{ 
\height{\MB}{b_{j}^{k}} - \height{\MB}{b_{j}^{i}} \geq k-i
}
\end{lemma}
\pf{ 
Let \(k \geq 0\). Since the miner of \(b_{j}^{k+1}\) has heard of \(b_{j}^{k}\), there is an edge from \(b_{j}^{k}\) to \(b_{j}^{k+1}\). The conclusion follows from Lemma~\ref{lem: XDAG-path-lower-bound}.
}

\begin{definition}[Forward Unheard]
Let \( k \geq j \geq 0\). The Forward Unheard for block \(b_k\) with respect to block \(b_j\) is denoted \(\FU{j}{k}\) and defined as the number of consecutive \(j\dash\FRSH\) blocks that the miner of \(b_k\) has not heard of, going backwards along the \(j\dash\FRSH\) sequence from the last such block mined before \(b_k\). If the miner of \(b_k\) has not heard of any \(j\dash\FRSH\) block, then we toss independent biased coins (with failure probability equal to the probability of message loss) and continue to increment the count until a success is encountered.
\end{definition}

\begin{remark}\label{rem: FU-Dist}
Let \(0 \leq j \leq k\). Let \(d\) be the probability of message loss. The random variable \( \FU{j}{k} \) has the same distribution as \(\mathsf{Geom}\pbr{1-d} - 1\), where \(\mathsf{Geom}(1-d)\) is a geometric random variable, with minimum value 1. Further, if \(k' \geq 0\) such that \(k' \neq k\), then \( \FU{j}{k} \) and \(\FU{j}{k'}\) are independent. 
\end{remark}

The intuition for defining \(\FU{j}{k}\) as above is illustrated through an example in Section~\ref{SS: Example}.  

\subsubsection{Backward Relative Special Honest} 
Relative to the \(j\)-th honest block \(b_j\), we define a sequence of `backward special (\(\BRSH\)) blocks as follows.

\begin{definition}[\(j\dash\BRSH\) Sequence]
Let \(j \geq 0\), and let \(b_j\) be the \(j\)-th honest block. The \(j\dash\BRSH\) sequence is a sequence of blocks \((b_{j}^{0},b_{j}^{-1},b_{j}^{-2},\cdots)\) such that \(b_j^0 = b_j\) and \(b_{j}^{-i}\) is the most recently mined block heard by \(b_{j}^{-(i-1)}\) for all \(i \geq 1\).
\end{definition}

\begin{figure}[t]
    \centering
    \includegraphics[width = 0.9 \textwidth]{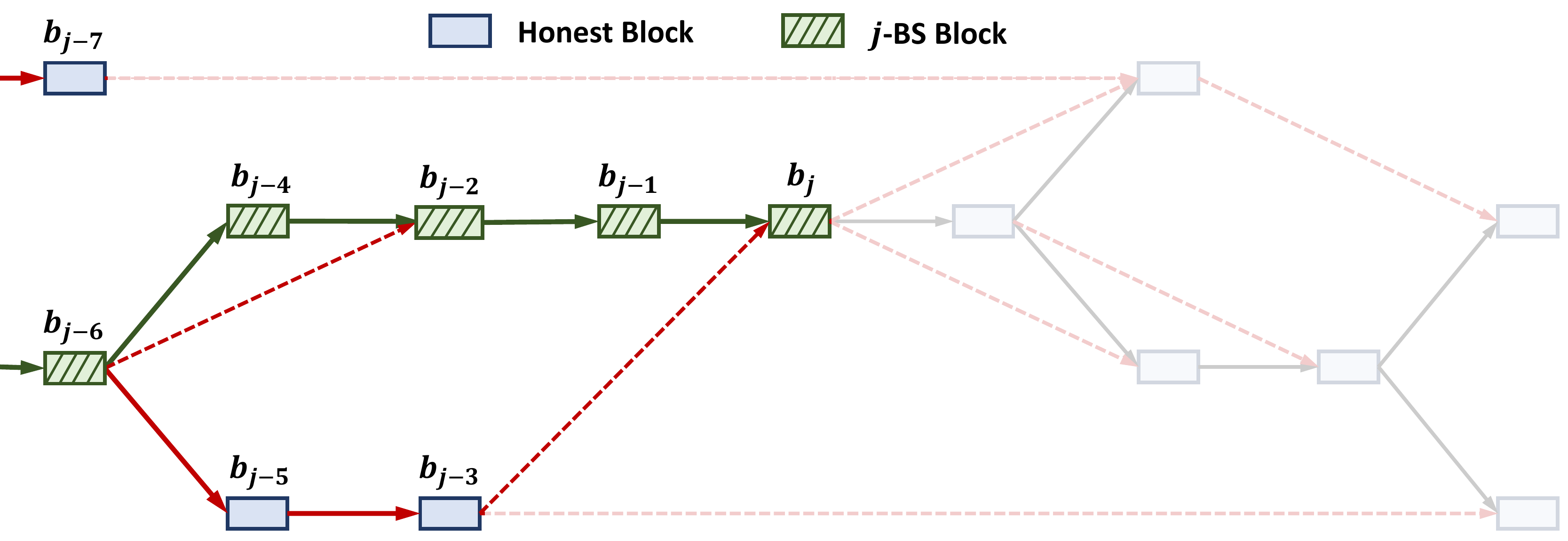}
    \caption{\(j\dash\BRSH\) Sequence}
    \label{fig: XDAG-BRSH}
\end{figure}

For \( 0 \leq i \leq j\), denote by \( \BS{j}{i} \) the number of \( j\dash\BRSH\) blocks mined between \(b_i\) and \(b_j\) (inclusive). Note that \( \BS{j}{i} \geq 1\), since \(b_j\) is always a \( j\dash\BRSH\) block. The \( j\dash\BRSH\) sequence associated with the transmission-graph in Figure~\ref{fig: XDAG} is shown in Figure~\ref{fig: XDAG-BRSH}. Notice that \( \BS{j}{j-5} = 4\) and \( \BS{j}{j-6} = 5\). This example is explored in more detail in Section~\ref{SS: Example}.

\begin{remark} \label{rem: BS-Dist}
Let \(0 \leq i \leq j\). Let \(d\) be the probability of message loss. The random variable \( \BS{i}{j}\) has the same distribution as \(1 + \sum_{k=1}^{j-i} \mathsf{Be}_k\pbr{1-d}\), where \(\mathsf{Be}_k\pbr{1-d}\) are i.i.d. Bernoulli random variables with success probability equal to \(1-d\). This is because \(b_j\) is a \(j\dash\BRSH\) block, and every previous block is independently \(j\dash\BRSH\) with probability \(1-d\).
\end{remark}

\begin{lemma} \label{lem: BRSH-forms-path}
Let \(j\geq 0\). The \(j\dash\BRSH\) sequence \((b_{j}^{0} , b_{j}^{-1} , b_{j}^{-2} , \cdots )\) is a backward directed path in transmission-graph. Further, if \(k > i \geq 0\), then the heights of the \(j\dash\BRSH\) blocks \( b_j^{-i}\) and \(b_{j}^{-k}\) satisfy:
\alns{ 
\height{\MB}{b_{j}^{-k}} - \height{\MB}{b_{j}^{-i}} \geq k-i
}
\end{lemma}
\pf{ 
Let \(i \geq 0\). Since the miner of \(b_{j}^{-(i-1)}\) has heard of \( b_{j}^{-i} \), there is an edge from \(b_{j}^{-i}\) to \(b_{j}^{-(i-1)}\). The conclusion follows from Lemma~\ref{lem: XDAG-path-lower-bound}.
}

\begin{definition}[Backward Unheard]
Let \(j \geq i \geq 0\). The Backward Unheard for block \(b_i\) with respect to block \(b_j\) is denoted \(\BU{j}{i}\) and defined as the number of consecutive \(j\dash\BRSH\) blocks whose miners have not heard of \(b_i\), going backwards along the \(j\dash\BRSH\) sequence from the first such block mined after \(b_i\). If none of the miners of \(j\dash\BRSH\) blocks have heard of \(b_i\), then we toss independent biased cones (with failure probability equal to the probability of message loss) and continue to increment the count until a success is encountered.
\end{definition}

\begin{remark}\label{rem: BU-Dist}
Let \(0 \leq i \leq j\). Let \(d\) be the probability of message loss. The random variable \( \BU{j}{i} \) has the same distribution as \(\mathsf{Geom}\pbr{1-d} - 1\), where \(\mathsf{Geom}(1-d)\) is a geometric random variable, with minimum value 1. Further, if \(i' \geq 0\) such that \(i' \neq i\), then \( \BU{j}{i} \) and \(\BU{j}{i'}\) are independent. 
\end{remark}

The intuition for defining \( \BU{j}{i}\) as above is illustrated through an example in Section~\ref{SS: Example}. The usefulness of these quantities is evident from Lemma~\ref{lem: MB-height-diff}.

\begin{lemma}\label{lem: MB-height-diff}
Let \(i < j < k\). 
\begin{enumerate}[label = (\roman*)]
    \item If \(\FS{j}{k-1} > \FU{j}{k} \), then \(   \height{\MB}{b_k} - \height{\MB}{b_j} \geq \FS{j}{k-1} - \FU{j}{k} \).
    \item If \(\BS{j}{i+1} > \BU{j}{i} \), then \(  \height{\MB}{b_j} - \height{\MB}{b_i} \geq \BS{j}{i+1} - \BU{j}{i} \).    
\end{enumerate}
\end{lemma}
    \input{Proofs/Pf-Height_Differential}

\subsection{An Example} \label{SS: Example}
Some of the concepts introduced above are best understood through an example. Consider the main-blocktree in Figure~\ref{fig: MB} and its associated transmission-graph in Figure~\ref{fig: XDAG}. The same transmission-graph is shown again in Figure~\ref{fig: Example}, where the \(j\dash\FRSH\) and \(j\dash\BRSH\) blocks are highlighted. Recall that in the transmission-graph, a directed edge from \(b_j\) to \(b_k\) indicates that the miner of \(b_k\) has heard of the block \(b_j\).

\paragraph{\(j\dash\FRSH\) Sequence} By definition, \(b_j\) is a \(j\dash\FRSH\) block. The miner of the next block \(b_{j+1}\) has heard of the most recent \(j\dash\FRSH\) block \(b_{j}\), so \(b_j^1 = b_{j+1}\) is a \(j\dash\FRSH\) block. The miner of \(b_{j+2}\) has heard of \(b_j^1\), so we have that \(b_{j+2} = b_{j}^{2}\) is also a \(j\dash\FRSH\) block. However, the miners of \(b_{j+3}\), \(b_{j+4}\) and \(b_{j+5}\) have not heard of \(b_j^2\), so these blocks are not \(j\dash\FRSH\). Finally, the miner of \(b_{j+6}\) has heard of \(b_j^2\), so we have that \(b_j^3= b_{j+6}\). Notice that no two \(j\dash\FRSH\) blocks can share the same height.

\paragraph{Forward Unheard} Let us consider \( \FU{j}{k}\) with \(k = j+4\). Starting from \(b_k\), we count the number of consecutive \(j\dash\FRSH\) blocks that the miner of \(b_k\) has not heard of, going backwards in the \(j\dash\FRSH\) sequence. The \(j\dash\FRSH\) sequence at the mining time of \(b_k\) is \(\pbr{b_j, b_{j+1}, b_{j+2}}\). Going backwards in this sequence, we see that the miner of \(b_k\) has not heard of \(b_{j+2}\), but has heard of \(b_{j+1}\). Therefore, we stop the count and have \(\FU{j}{j+4} = 1 \). 

\paragraph{\(j\dash\BRSH\) Sequence} By definition, \(b_j\) is a \(j\dash\BRSH\) block. Since the miner of the most recent \(j\dash\BRSH\) block (\(b_j\)) has heard of the previous block \(b_{j+1}\), we have that \(b_{j-1} = b_j^{-1}\) is a \(j\dash\BRSH\) block. Similarly, the miner of \(b_{j}^{-1}\) has heard of \(b_{j-2}\), so we have that \(b_{j-2} = b_{j}^{-2} \) is a \(j\dash\BRSH\) block. However, the miner of \(b_{j}^{-2}\) has not heard of \(b_{j-3}\), and so \(b_{j-3}\) is not a \(j\dash\BRSH\) block. Notice that no two \(j\dash\BRSH\) blocks can share the same height.

\paragraph{Backward Unheard} Let us consider \( \BU{j}{i}\) with \(i = j-3\). Starting from the first \(j\dash\BRSH\) block mined after \(b_i\), we count the number of consecutive \(j\dash\BRSH\) blocks  that have not heard of \(b_i\), going backwards in the \(j\dash\BRSH\) sequence. The \(j\dash\BRSH\) sequence at the mining time of \(b_i\) is \(b_j, b_{j-1}, b_{j-2}\). Going backwards in this sequence, we see that the miners of \(b_{j-2}\) and \(b_{j-1}\) have not heard of \(b_{i}\), but the miner of \(b_{j}\) has. Therefore, we stop the count and have \(\BU{j}{j-3} = 2\).

\begin{figure}
    \centering
    \includegraphics[width = 0.9 \textwidth]{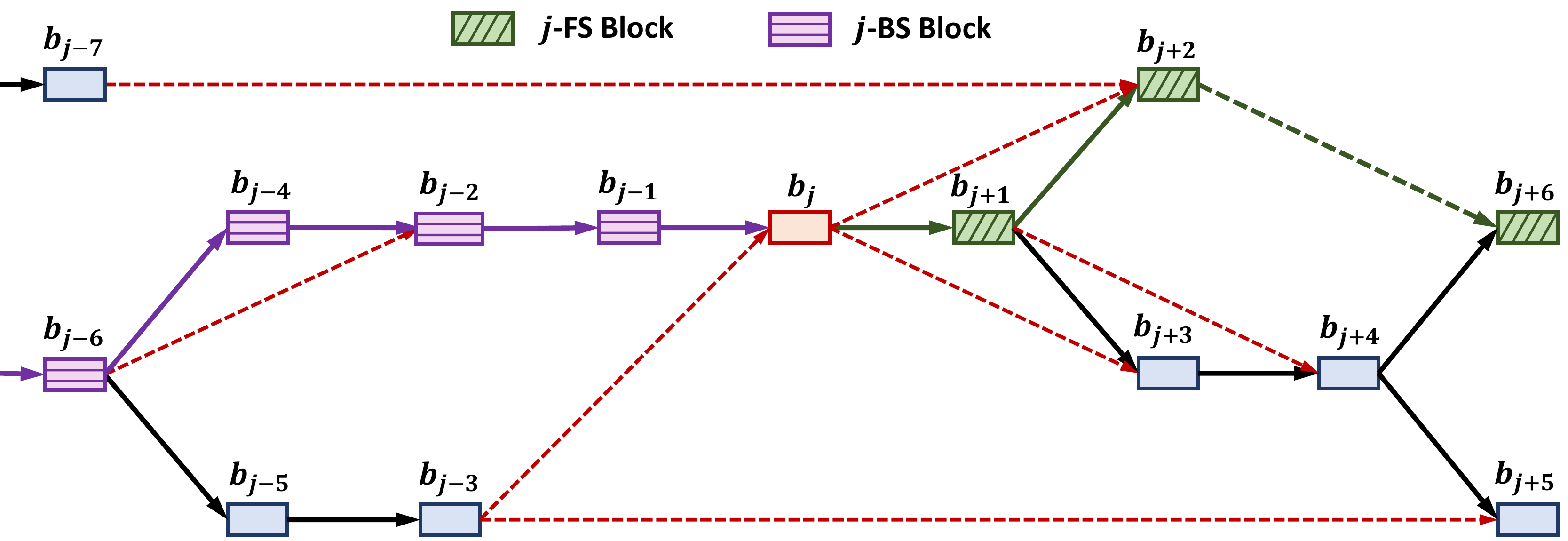}
    \caption{Example}
    \label{fig: Example}
\end{figure}

\subsection{Catch-up Events and \(\ss\)-Nakamoto Blocks} \label{SS: Nakamoto}
In this section, we define catch-up events and \(\ss\)-Nakamoto blocks. 

\paragraph{Adversarial arrivals} Let \(\A{b_i}{b_j}\) denote the number of adversarial blocks mined between the mining times of the \(i\)-th and \(j\)-th honest blocks. Similarly, let \(\A{b_j^i}{b_j^k}\) denote the number of adversarial blocks mined between the \(i\)-th and \(k\)-th \(j\dash\FRSH\) blocks. 

\begin{definition}[Catch-up Events]
Let \(0 \leq i < j < k \), and let \( 0 < \ss \leq 1 \). The forward and backward catch-up events are respectively defined as:
\aln{ 
\Bf{j}{k}{\ss} &\colon \A{b_j}{b_k} \geq \ss \cdot \FS{j}{k-1}  - \FU{j}{k} \label{eq: Bf}\\
\Bb{i}{j}{\ss} &\colon \A{b_i}{b_j} \geq \ss \cdot \BS{j}{i+1} - \BU{j}{i} \label{eq: Bb}
}
\end{definition}
These events are catch-up events in the following sense. If the event \(\Bf{j}{k}{\ss}\) occurs for some \(j < k\) and \( 0 < \ss \leq 1\), then more adversarial blocks have been mined in \( [\tau_j,\tau_k]\) than \emph{effective} \( j \dash \FRSH\) blocks. For instance, when \(\ss = 1\), more adversarial blocks have been mined in \(\sbr{\tau_j,\tau_k} \) than the number of blocks in the \(j\dash\FRSH\) sequence from \(b_j\) up to the last \(j\dash\FRSH\) block heard by the miner of \(b_k\). If these adversarial blocks were to form a side chain rooted at \(b_j\), then there is a possibility that \(b_k\) mines on this side chain. However, such an attack would fail if the catch-up events did not occur, because there would necessarily be a chain longer than the adversarial side chain that the miner of \(b_k\) is aware of. Here, \( \ss \) is a robustness measure: If the catch-up event does not occur for a small value of \(\ss\), then the honest blocks have a considerable lead over adversarial side chains.

\paragraph{Longest Chains of the main-blocktree} A longest chain of the main-blocktree at time \(t\) is a path in \(\MBt\) whose length is no shorter than any path in \(\MBt\). Notice that \(\MBt\) can have multiple longest chains.

\paragraph{\(\ss\)-Nakamoto Blocks}
A desirable property for honest block \(b_j\) is the existence of some \(0 < \eta \leq 1\) such that none of the events \( \Bf{j}{k}{\ss}\) occur for any \(i < j\) and none of the events \( \Bb{i}{j}{\ss}\) occur for any \( k > j\). Intuitively, this ensures that adversarial chains rooted at any block \(b_i\) with \(i < j\) are never long enough for any block \(b_k\) with \(k > j\) to extend them. In turn, this suggests that block \(b_j\) would be in every longest chain of the main-blocktree right from being mined. We formalize this property using the notion of \(\ss\)-Nakamoto blocks below, and make this intuition rigorous in Section \ref{S: 05-Results}.

\noindent 
\begin{definition}[\(\ss\)-Nakamoto Block]
Let \(j \geq 0\) and \( 0 < \ss \leq 1\). The \(j\)-th honest block \(b_j\) is said to be an \(\ss\)-Nakamoto block if the event
\aln{ 
\Njss \colon \sbr{ \bigcap_{i\colon i < j} \Bbc{i}{j}{\ss}}  \bigcap \sbr{ \bigcap_{k\colon k > j} \Bfc{j}{k}{\ss}} \label{eq: def-Nakamoto}
}
occurs. 
\end{definition}
We remark that an \( \ss\)-Nakamoto block is also an \( \ss' \)-Nakamoto block, for any \(\ss' \in [\ss,1]\).

\subsection{User-unheard-criterion} \label{SS: Observer}
\(\ss\)-Nakamoto blocks are useful because they belong to every longest chain of the main-blocktree. However, this is not equivalent to \(\ss\)-Nakamoto blocks belonging to any \textit{user's} chain for all future time. This is because at any given time, a given user has not necessarily heard of all the blocks in \(\MBt\). In this section, we introduce the tools that are relevant to analyzing the state of a user's chain with respect to the main-blocktree. We begin by introducing `Unheard' with respect to a user.   

\begin{definition}[User-unheard, \( \Un{h}{j}{k}\)]
Let \(j\geq 0\) and let \(b_j\) be the \(j\)-th honest block. Let \((b_j^0,b_j^1,\cdots) \) denote the \(j\dash\FRSH\) sequence. Let \(h\) be an honest user. For any \(k \geq 0\), we define \( \Un{h}{j}{k}\) as the number of consecutive \(j\dash\FRSH\) blocks that the user \(h\) has not heard of, going backwards along the \(j\dash\FRSH\) sequence from \(b_j^k\). If \(h\) has not heard of any \(j\dash\FRSH\) block, then we toss independent biased coins (with failure probability equal to the probability of message loss) and continue to increment the count until a success is encountered.
\end{definition}

\begin{remark}\label{rem: Unheard-Dist}
Let \(j,k \geq 0\). Let \(h\) be any honest user and \(d\) be the probability of message loss. Let \(\mathsf{Geom}(1-d)\) be a geometric random variable, with minimum value 1. The distribution of the random variable \( \Un{h}{j}{k} \) depends on \(h\):
\begin{itemize}
    \item If \(h\) has not mined any block after the mining time of \(b_j\), then \(\Un{h}{j}{k} \) has the same distribution as \(\mathsf{Geom}\pbr{1-d} - 1\).
    \item If \(h\) is a miner of a \(j\dash\FRSH\) block, then \(h\) has heard of its own block as well as the \(j\dash\FRSH\) block that immediately preceded it. However, all other delays from the miners of \(j\dash\FRSH\) blocks to \(h\) are still independent and identically distributed. In this case, \( \Un{h}{j}{k}\) is stochastically dominated by \(\mathsf{Geom}(1-d)\).
    \item If \(h\) is a miner of a block \(b_{\ell}\) with \(\ell > j\), such that \(b_{\ell}\) is not a \(j\dash\FRSH\) block, then the delay from the most recent \(j\dash\FRSH\) block before \(b_{\ell}\) to \(h\) is infinity. However, all other delays from the miners of \(j\dash\FRSH\) blocks to \(h\) are still independent and identically distributed. In this case, \(\Un{h}{j}{k}\) is stochastically dominated by \(\mathsf{Geom}(1-d)\).
\end{itemize}
In all cases, we have \( \Un{h}{j}{k} \leq \mathsf{Geom}\pbr{1-d}\), where the \(\leq\) sign indicates stochastic domination.
\end{remark}

We compare `forward-unheard' (denoted \(\FU{j}{k}\)) and `user-unheard' (denoted \( \Un{h}{j}{k}\)). Although similar in spirit, the quantity \(\FU{j}{k}\) counts the number of consecutive \(j\dash\FRSH\) blocks not heard by the \textit{miner} of block \(b_k\), going backwards in the \(j\dash\FRSH\) sequence from the last such block before \(b_k\), whereas the quantity \( \Un{h}{j}{k}\) fixes an honest user \(h\) and similarly counts the number of consecutive such blocks not heard by \(h\). Since \(b_k\) and \(b_{\ell}\) are mined by two different miners, then \( \FU{j}{k}\) and \(\FU{j}{\ell}\) are independent random variables. However, \( \Un{h}{j}{k}\) and \(\Un{h}{j}{\ell} \) need not be independent. 

Next, we introduce the user-unheard-criterion, which will later allow us to infer useful information about the state of a user's chain from the main-blocktree. 

\begin{definition}[User-unheard-criterion]
Let \(h\) be an honest user. Let \(j \geq 1\), \(k_0 \geq 1\), and \(0 < \ss < 1\). We say that the \( \pbr{h,j,\ss,k_0}\)-user-unheard-criterion is satisfied if
\aln{
\Un{h}{j}{k} < \pbr{\frac{1-\ss}{2}}  k,  \quad \forall k \geq k_0, \label{eq: Observer_Criterion}
}
\end{definition}

%% file: Proofs/Pf-Height_Differential.tex
\pf{ 
We prove statement (i). Let \( m = \FS{j}{k-1}  - \FU{j}{k} > 0\). Consider the \(j\dash\FRSH\) sequence \( (b_j^0, b_j^1,\cdots,b_j^m,\cdots )\). Note that \(\FS{j}{k-1} \) is the number of \(j\dash\FRSH\) blocks mined before \(b_k\), and \(\FU{j}{k}\) is the number of consecutive blocks from this sequence that were not heard by \(b_k\) going backward. Therefore, \(b_k\) has heard of \(b_j^m\), and will mine at a greater height. Since the \( j\dash\FRSH\) blocks are all mined at different heights, we have
\alns{
\height{\MB}{b_k} - \height{\MB}{b_j} >  \height{\MB}{b_{j}^{m}} - \height{\MB}{b_j^0} \geq m = \FS{j}{k-1} - \FU{j}{k},
}
as desired. The proof of statement (ii) is essentially the same, because the forward sequences \(\FRSH_j\) and \( \FRU_j\) map to the backward sequences \(\BRSH_j\) and \(\BRU_j\) under reversing the directions of the edges in the transmission-graph.
}

%% file: 05_Results.tex
In this section, we present an outline of the proof of Theorem~\ref{thm: Main}. First, deterministic and probabilistic results are stated. These results are used as building blocks in the proof sketch of our main result, Theorem~\ref{thm: Main} which is presented in Section~\ref{SS: Sketch}. Rigorous proofs of all the results, including the main result are relegated to the appendix.

\subsection{Deterministic Results}
\begin{theorem} \label{thm: Nakamoto-implies-longest-chain}
    \input{Theorem-Statements/Thm-Nakamoto-implies-Longest-Chain}
\end{theorem}
\pf{ 
See Appendix~\ref{A: Nakamoto-implies-longest-chain}
}

\begin{theorem} \label{thm: Obs-and-Nak}
    \input{Theorem-Statements/Thm-Deterministic-Observer}
\end{theorem}
\pf{ 
See Appendix~\ref{A: Obs-and-Nak}. 
}

\subsection{Probabilistic Results}

\begin{lemma} \label{lem: B-are-independent}
    \input{Theorem-Statements/Thm-Independence-of-B}
\end{lemma}
\pf{ 
See Appendix~\ref{A: B-are-independent}.
}

\begin{theorem} \label{thm: Lower-bound-on-Nakamoto}
    \input{Theorem-Statements/Thm-Probability-Lower-Bound}
\end{theorem}
\pf{ 
See Appendix~\ref{A: Lower-bound-on-Nakamoto}.
}

\begin{lemma} \label{lem: Probability-of-catchup-events}
    \input{Theorem-Statements/Thm-Probability-of-Catchups}
\end{lemma}
\pf{ 
See Appendix~\ref{A: Probability-of-catchup-events}.
}

\begin{theorem} \label{thm: Exp-sqrt-decay}
    \input{Theorem-Statements/Thm-Sqrt-Decay}
\end{theorem}
\pf{ 
See Appendix~\ref{A: Exp-sqrt-decay}.
}

\begin{theorem} \label{thm: Bootstrap}
    \input{Theorem-Statements/Thm-Bootstrap}
\end{theorem}
\pf{ 
See Appendix~\ref{A: Bootstrap}.
}

\begin{theorem} \label{thm: Probability-Observer-Criterion}
    \input{Theorem-Statements/Thm-Probabilistic-Observer}
\end{theorem}
\pf{ 
See Appendix~\ref{A: Probability-Observer-Criterion}.
}

\subsection{Proof Sketch} \label{SS: Sketch}
\(\ss\)-Nakamoto blocks are special blocks that are part of every longest chain in the main-blocktree, for all time after they are mined (Theorem~\ref{thm: Nakamoto-implies-longest-chain}). Therefore, if a transaction is included in an \(\ss\)-Nakamoto block or in any of its ancestors, then it will be included in every longest chain of the main-blocktree. Furthermore, if an honest user is up-to-date with the main-blocktree (specifically, by satisfying a relevant user-unheard-criterion, and therefore having heard of more forward special blocks with respect to the \(\ss\)-Nakamoto block, than adversarial blocks), then it is reasonable to expect that the honest user will also include the transaction in its chain. This idea is formalized in the user-unheard-criterion, and made rigorous in Theorem~\ref{thm: Obs-and-Nak}.

The core idea behind the proof of the main result (Theorem~\ref{thm: Main}) is to show that \(\ss\)-Nakamoto blocks occur frequently with high probability (Theorem~\ref{thm: Bootstrap}). In turn, this implies that a transaction is highly likely to be included in an \(\ss\)-Nakamoto block or its ancestor, soon after it is made. Theorem~\ref{thm: Probability-Observer-Criterion} then shows that it is also highly likely that a relevant user-unheard-criterion will hold for the honest user, once a waiting time has elapsed. 

To prove Theorem~\ref{thm: Bootstrap}, we bootstrap from a milder version of it, which is stated as Theorem~\ref{thm: Exp-sqrt-decay}. Inspired by the proof strategy in~\cite{DKT_2020}, we prove Theorem~\ref{thm: Exp-sqrt-decay} by separating catch-up events into long and short-term catch-up events. We show that long term catch-ups occur rarely as a consequence of Lemma~\ref{lem: Probability-of-catchup-events}, and short-term catch-up probabilities are bounded using a lower bound on the probability that the \(j\)-th block is an \(\ss\)-Nakamoto block (Theorem~\ref{thm: Lower-bound-on-Nakamoto}), and the fact that non-overlapping catch-up events are independent (Lemma~\ref{lem: B-are-independent}). This establishes that a transaction is highly likely to be included in the main-blocktree for all future time, once a waiting period has elapsed. This proves thereom~\ref{thm: Bootstrap}.

Next, we show that an honest user is likely to always be up-to-date with the main-blocktree, forever after a waiting time. This is done by explicitly bounding the probability that a relevant user-unheard-criterion is violated using tools from stochastic analysis. This proves Theorem~\ref{thm: Probability-Observer-Criterion}.  

Theorems~\ref{thm: Bootstrap} and~\ref{thm: Probability-Observer-Criterion} together imply the main result. The details of this proof are presented in Appendix~\ref{A: Main}.

%% file: Theorem-Statements/Thm-Nakamoto-implies-Longest-Chain.tex
Let \(j\geq 1\) and \( 0 < \ss \leq 1\). Let \(b_j\) be the \(j\)-th honest block, and \(\tau_j\) be the mining time of \(b_j\). Suppose \(b_j\) is an \(\ss\)-Nakamoto block. Then:
\begin{enumerate}[label = (\roman*)]
\item \(b_j\) is the unique honest block at its height in \(\MBt\) for all \(t \geq \tau_j\).
\item \(b_j\) is in every longest chain of \(\MBt\) for all \(t \geq \tau_j\). 
\item \(b_j\) is the unique block at its height in \(\MB(\tau_j)\).
\item For any \(k\) such that \(k > j\), \(b_k\) is a descendant of \(b_j\) in the main-blocktree.
\end{enumerate}

%% file: Theorem-Statements/Thm-Deterministic-Observer.tex
Let \(h\) be an honest user. Let \(j \geq 1\) and \( 0 < \ss < 1\). Let \(k_0 =  \ceil{\frac{2\ss}{1-\ss}} \). Let \(\Chain{h}{t}\) denote the chain held by \(h\) at time \(t\). If \(b_j\) is an \(\ss\)-Nakamoto block, and if the \(\pbr{h,j,\ss, k_0}\)-user-unheard-criterion is satisfied, then \(b_j \in \Chain{h}{t}\) for all \( t \geq \tau_j^{k_0}\), where \( \tau_j^{k_0} \) is the mining time of the \(k_0\)-th \(j\dash\FRSH\) block. 

%% file: Theorem-Statements/Thm-Independence-of-B.tex
Let \(0 \leq i < j \leq j' < k\), and \(0 \leq \ss \leq 1\). The events:
\begin{itemize}
    \item \( \Bb{i}{j}{\ss}\) and \( \Bf{j'}{k}{\ss}\) are independent.
    \item \( \Bf{i}{j}{\ss} \) and \( \Bb{j'}{k}{\ss} \) are independent.
\end{itemize}

%% file: Theorem-Statements/Thm-Probability-Lower-Bound.tex
Let \(j \geq 1\) and \(0 \leq \ss \leq 1\). Recall that \(b_j\) is the \(j\)-th honest block and \(\Njss\) is the event that \(b_j\) is an \(\ss\)-Nakamoto block. If \( \frac{\beta}{1-\beta} < \ss \cdot \pbr{1-d}\), then there exists a positive constant \(p_0 > 0\) such that
\alns{ 
\P{\Njss} \geq p_0 > 0
}

%% file: Theorem-Statements/Thm-Probability-of-Catchups.tex
Let \( 0 < \ss \leq 1\) and \(i < j < k\). Let \( \Bf{j}{k}{\ss} \) and \( \Bb{i}{j}{\ss} \) be the catch-up events defined in~\eqref{eq: Bf},~\eqref{eq: Bb}. If \( \frac{\beta}{1-\beta} < \ss \cdot \pbr{1-d}\), there exists a constant \(c > 0\) such that
\alns{ 
\P{ \Bf{j}{k}{\ss} } &\leq e^{-c \pbr{ k-j} } \\
\P{ \Bb{i}{j}{\ss} } &\leq e^{-c \pbr{ j-i} } 
}

%% file: Theorem-Statements/Thm-Sqrt-Decay.tex
Let \( 0 < \ss \leq 1\). Let \(\beta\) be the fraction of computational power in the system that is adversarial and \(d\) be the probability of message loss. Let \(B_{s,s+t}^{\pbr{\ss}}\) be the event that there are no \(\ss\)-Nakamoto blocks in \([s,s+t]\). If \( \frac{\beta}{1-\beta} < \ss\cdot  \pbr{1-d}\), then there exists a constant \(c_0 > 0 \) such that for any \(s,t \geq 0\),
\alns{ 
\P{B_{s,s+t}^{\pbr{\ss}}} \leq e^{-c_0 \sqrt{t}}
}

%% file: Theorem-Statements/Thm-Bootstrap.tex
Let \( 0 < \ss \leq 1\). Let \(\beta\) be the fraction of computational power in the system that is adversarial and \(d\) be the probability of message loss. Let \(B_{s,s+t}^{\pbr{\ss}}\) be the event that there are no \(\ss\)-Nakamoto blocks in \( \sbr{s,s+t}\). If \( \frac{\beta}{1-\beta} < \ss \cdot \pbr{1-d}\), then for every \( \eps > 0 \), there exist positive constants \(A\), \(a\) such that for any \(s,t > 0\),
\alns{ 
\P{B_{s,s+t}^{\pbr{\ss}}} \leq A \exp \pbr{-a t^{1-\eps}}.
}

%% file: Theorem-Statements/Thm-Probabilistic-Observer.tex
Let \(0 < \ss < 1\). Suppose the fraction \(\beta\) of computational power in the system that is adversarial, and the probability \(d\) of message loss satisfies \(\frac{\beta}{1-\beta} < \ss \cdot \pbr{1-d}\). Given \(s \geq 0\), let \(b_J\) be the first \(\ss\)-Nakamoto block mined after time \(s\). There exist constants \(C\), \(c > 0\) such that for any honest user \(h\) and for all \(k' \geq 1 \), 
\alns{ 
\P{ \pbr{h,J,\ss,k'}\dash\text{user-unheard-criterion fails} } \leq  Ce^{-c k'}
}

%% file: 06_Conclusion.tex
In this work, we introduced the \(0\dash\infty\) model: a framework to study the impact of random message losses on the security of the proof-of-work longest-chain protocol. We investigated the security of this protocol by analyzing the transmission-graph, a dynamically evolving graph that captures the delays incurred by the blocks mined by honest miners. Specifically, we studied special sequences of blocks and identified random variables associated with them that are amenable to analysis. These random variables were used to define useful objects and desirable events, such as \(\ss\)-Nakamoto blocks and user-unheard-criterion respectively. These ideas allowed us to generalize analysis techniques from the synchronous delay model to a setting where delays are possibly infinite. We showed that the condition \(\frac{\beta}{1-\beta} < 1-d\) is sufficient for a transaction to satisfy desired security properties except with a probability that decays almost exponentially in the security parameter. This greatly improved the known threshold of the fraction of adversarial power that is tolerable for a given probability of message loss in an instance of point-to-point communication. 

%% file: 07_Appendix.tex
    \section{Proofs of Deterministic Results}
    \subsection{Proof of Theorem~\ref{thm: Nakamoto-implies-longest-chain}} \label{A: Nakamoto-implies-longest-chain}
    \input{Proofs/Pf-Nakamoto_Implies_Longest_Chain}


    \subsection{Proof of Theorem~\ref{thm: Obs-and-Nak}} \label{A: Obs-and-Nak}
    \input{Proofs/Pf-Deterministic-Observer}

    \section{Proofs of Probabilistic Results}
    \subsection{Proof of Lemma~\ref{lem: B-are-independent}} \label{A: B-are-independent}
    \input{Proofs/Pf-Independence_of_B}
    
    \subsection{Proof of Theorem~\ref{thm: Lower-bound-on-Nakamoto}} \label{A: Lower-bound-on-Nakamoto}
    \input{Proofs/Pf-Probability_Lower_Bound}

    \subsection{Proof of Lemma~\ref{lem: Probability-of-catchup-events}} \label{A: Probability-of-catchup-events}
    \input{Proofs/Pf-Probability_of_Catchups}

    \subsection{Proof of Theorem~\ref{thm: Exp-sqrt-decay}} \label{A: Exp-sqrt-decay}
    \input{Proofs/Pf-Sqrt_Decay}

    
    \subsection{Proof of Theorem~\ref{thm: Bootstrap}} \label{A: Bootstrap}
    \input{Proofs/Pf-Bootstrap}

    \subsection{Proof of Theorem~\ref{thm: Probability-Observer-Criterion}} \label{A: Probability-Observer-Criterion}
    \input{Proofs/Pf-Probabilistic-Observer}

    \subsection{Proof of Theorem~\ref{thm: Main}} \label{A: Main}
    \input{Proofs/Pf-Main-Result}

%% file: Proofs/Pf-Nakamoto_Implies_Longest_Chain.tex
\begin{reptheorem}{thm: Nakamoto-implies-longest-chain}
    \input{Theorem-Statements/Thm-Nakamoto-implies-Longest-Chain}
\end{reptheorem}
\pf{ 
Let \(\Njss\) be the event that \(b_j\) is an \(\ss\)-Nakamoto block. Recall that
\alns{ 
\Njss = \sbr{ \bigcap_{i\colon i < j} \Bbc{i}{j}{\ss}}  \bigcap \sbr{ \bigcap_{k\colon k > j} \Bfc{j}{k}{\ss}}
}
Suppose \(\Njss\) occurs, so that \(b_j\) is an \(\ss\)-Nakamoto block. 
\begin{enumerate}[label = (\roman*)]
    \item Let \(t \geq \tau_j\). We show that \(b_j\) is the unique honest block at its height in \( \MBt\), i.e. for any \(i,k \geq 0\) such that \(i < j < k\), we show that
    \aln{ 
        \height{\MB}{b_{i}} < \height{\MB}{b_{j}} < \height{\MB}{b_k}. \label{eq: height-monotonicity}
    }
    We begin by proving the first inequality in~\eqref{eq: height-monotonicity}. Since \(\Njss\) occurs, it follows that the event \(\Bbc{i}{j}{\ss}\) occurs. Therefore, we have
    \alns{ 
        \height{\MB}{b_j} - \height{\MB}{b_i} &\stackrel{(a)}{\geq} \BS{j}{i+1} - \BU{j}{i} \\
        &\stackrel{(b)}{\geq} \ss \cdot \BS{j}{i+1} - \BU{j}{i} \\
        &\stackrel{(c)}{>} \A{b_i}{b_j},  \label{eq: Evj1-implication}
    }
    where (a) is from Lemma~\ref{lem: MB-height-diff}, (b) follows from the fact that \( \ss \leq 1\), and (c) is the definition of \(\Bbc{i}{j}{\ss}\). Since \(\A{b_i}{b_j} \geq 0\), it follows that \(\height{\MB}{b_i} < \height{\MB}{b_j}\). Similarly, the occurrence of \( \Bfc{j}{k}{\ss}\) implies that \(\height{\MB}{b_j} < \height{\MB}{b_k}\). Thus, \(b_j\) is the unique honest block at its height in \(\MBt\).

    \item We show that \(b_j\) is in every longest chain of \(\MBt\). Let \(\mathcal{V}\) be any longest chain in \(\MBt\), i.e.
    \[
    \mathcal{V} = v_0 - v_1 - \cdots - v_{\height{\MB}{b_j}} - \cdots - v_{m-1} - v_{m},
    \]
    where \(v_0\) is the genesis block, \(v_{\height{\MB}{b_j}} =: v\) is the block at the height of \(b_j\), and \(m\) is the height of \(\MBt\). It suffices to show that \(v = b_j\). In fact, since \(b_j\) is the unique honest block at its height, it suffices to only show that \(v\) must be an honest block. Starting from \(v\) and traversing blocks in \( \mathcal{V}\) in the backward direction, let \( b_{i}  \) denote the first honest block encountered in \(\mathcal{V}\), not including \(v\) itself. Since the genesis block is honest, such a \(b_i\) exists. Similarly, we also want to traverse blocks in \(\mathcal{V}\) along the forward direction. Here, we have two cases:
    
    \textbf{Case 1:} There is an honest block in \(\mathcal{V}\) after \(v\). In this case, let \(b_k\) be the first honest block in \(\mathcal{V}\) after \(v\). Let \(H\) denote \(\height{\MB}{b_k}\), and \(\mathcal{V}_1\) be the sequence \( b_{i} - \cdots - v - \cdots - b_{k}\). 
    
    \textbf{Case 2:} There is no honest block in \(\mathcal{V}\) after \(v\). In this case, let \(b_k\) be the first honest block mined after time \(t\). Let \(H = m\) denote the height of the main-blocktree, and \(\mathcal{V}_1\) be the sequence \( b_{i} - \cdots - v - \cdots - v_m\). 

    In both cases, let \(N\) denote the number of blocks in \(\mathcal{V}_1\), not including \(b_i\). Notice that the blocks are at consecutive heights. We have:     
    \alns{ 
        N &= H - \height{\MB}{b_{i}}  \\
        & \stackrel{(a)}{\geq} \height{\MB}{b_j} + \FS{j}{k-1} - \FU{j}{k} - \height{\MB}{b_i} \\
        & \stackrel{(b)}{\geq}  \pbr{ \FS{j}{k-1} - \FU{j}{k}} +\pbr{ \BS{j}{i+1} - \BU{j}{i}}  
        \\
        & \stackrel{(c)}{\geq} \pbr{\A{b_i}{b_j} + 1 } + \pbr{ \A{b_j}{b_k} + 1 }  
        \\
        & > \A{b_i}{b_k} + 1,
    }
    where (a) follows from Lemma~\ref{lem: MB-height-diff} in case 1, and the fact that \(j\dash\FRSH\) blocks are mined at different heights in case 2. Further, (b) is due to Lemma~\ref{lem: MB-height-diff} and the fact that \(\ss \leq 1\), and (c) follows because the events \(\Bbc{i}{j}{\ss} \) and \( \Bfc{j}{k}{\ss}\) occur. However, by definition of \(b_i\) and \(b_k\), all blocks in \(\mathcal{V}_1\) after \(b_i\) except for \(b_k\) and possibly \(v\) are adversarial. Therefore, there must be at least one honest block in \(\mathcal{V}_1\) strictly between \(b_i\) and \(b_k\). Since the only possibility for this is \(v\), it follows that \(v\) must be an honest block. Since \(b_j\) is the unique honest block at the height of \(v\), we conclude that \(b_j\) in every longest chain of \(\MBt\), as desired.
    
    \item Let \(v\) be a block in \(\MB(\tau_j)\) such that \( \height{\MB}{v} = \height{\MB}{b_j}\). Let \(\mathcal{T}\) be the tine of ancestors of \(v\). Traversing blocks along the backward direction starting from \(v\), let \(b_i\) denote the first honest block encountered. The portion of \(\mathcal{T}\) between \(b_i\) and \(v\) consists of only adversarial blocks. However,
    \alns{ 
    \height{\MB}{v} - \height{\MB}{b_i} &= \height{\MB}{b_j} - \height{\MB}{b_i} \\
    & \stackrel{(a)}{\geq} \ss \cdot \BS{j}{i+1} - \BU{j}{i} \\
    & \stackrel{(b)}{>} \A{b_i}{b_j},
    }
    where (a) is from Lemma~\ref{lem: MB-height-diff}, and (b) is true because \(b_j\) is an \(\ss\)-Nakamoto block. Since \(v\) is the only possibility for an honest block in the portion of \(\mathcal{T}\) after \(b_i\), it follows that \(v\) must be an honest block. From (i), we conclude that \(v = b_j\). This concludes the proof. 
    
    \item Let \(k > j\). Since \(b_j\) is an \(\ss\)-Nakamoto block, we have from the proof of (i) that \( \height{\MB}{b_k} > \height{\MB}{b_j}\).  Let \(v\) denote the ancestor of \(b_k\) at the height of \(b_j\) in the main-blocktree. Starting from \(v\) and traversing blocks along the ancestors of \(v\), let \(b_i\) denote the first honest block encountered. This is exactly case 1 in the proof of (ii), and it follows that \(v = b_j\).
\end{enumerate}
}

%% file: Proofs/Pf-Deterministic-Observer.tex
\begin{reptheorem}{thm: Obs-and-Nak}
    \input{Theorem-Statements/Thm-Deterministic-Observer}
\end{reptheorem}
\pf{ 
Let \(|\Chain{h}{t}|\) denote the number of blocks in \( \Chain{h}{t}\). Let \(\tau_j\) denote the mining time of \(b_j\), and let \(\tau_{j}^{k}\) denote the mining time of \(b_j^k\). Since \(b_j\) is an \(\ss\)-Nakamoto block, we know that the event \(\Bfc{j}{k}{\ss}\) occurs for all \(k > j\). In turn, this implies
\alns{ 
\A{b_j^0}{b_{j}^{k+1}} < \ss \cdot \pbr{k+1} \quad \forall k > 0. 
}
However,  \(\ss \pbr{k+1} \leq \pbr{\frac{1+\ss}{2}}k\) whenever \( k \geq k_0\). Further, since \(\pbr{h,j,\ss, k_0} \)-user-unheard-criterion holds, we have that \(\Un{h}{j}{k} < \pbr{\frac{1-\ss}{2}}k \) whenever \(k \geq k_0\). These facts imply:
\aln{ 
\A{b_j^0}{b_{j}^{k+1}} < k - \Un{h}{j}{k} \quad \forall k \geq k_0. \label{eq: Observer_Criterion_Implication}
}
Fix \(k \geq k_0\). Let \(t\) be such that \( \tau_j^k \leq t < \tau_j^{k+1}\). We show that \(h\) includes \(b_j\) in its chain at time \(t\). 

Since \(\A{b_j^0}{b_j^{k+1}} \geq 0 \), it follows from~\eqref{eq: Observer_Criterion_Implication} that \( \Un{h}{j}{k} < k\). Therefore, \(h\) has heard of at least one \(j\dash\FRSH\) block. Therefore, \( |\Chain{h}{t}| \geq \height{\MB}{b_j} \). Let \(v \in \Chain{h}{t}\) be the block at the height of \(b_j\). To show that \(b_j \in \Chain{h}{t}\), it suffices by statement (iv) in Theorem~\ref{thm: Nakamoto-implies-longest-chain} to prove that \(\Chain{h}{t}\) contains at least one honest block mined in \( [\tau_j^k, \infty)\). 

Let \(\Chain{h}{t}^{\lceil v}\) denote the sub-chain of \(\Chain{h}{t}\) starting from \(v\), i.e. \(\Chain{h}{t}^{\lceil v}\) contains blocks of \( \Chain{h}{t} \) that are at height no less than that of \(v\). We show that~\eqref{eq: Observer_Criterion_Implication} implies that \(\Chain{h}{t}^{\lceil v}\) cannot contain all adversarial blocks. We know from statement (iii) of Theorem~\ref{thm: Nakamoto-implies-longest-chain} that \(b_j\) is the unique block at its height in \(\MB(\tau_j)\). Therefore, all blocks in \(\Chain{h}{t}^{\lceil v}\) are mined at or after time \(\tau_j\). Since \(t < \tau_j^{k+1}\), we have
\alns{ 
\A{b_j^0}{b_j^{k+1}} \stackrel{(a)}{<}  k - \Un{h}{j}{k} \stackrel{(b)}{\leq} |\Chain{h}{t}^{\lceil v}|, 
}
where (a) is the same as~\eqref{eq: Observer_Criterion_Implication} and (b) follows from the fact that \(h\) has heard of \(b_j^m\), where \(m = k - \Un{h}{j}{k}\). Therefore, \(h\) adopts a chain that has length at least \(\height{\MB}{v} + m \). We conclude that there exists an honest block \(b_k \in \Chain{h}{t}^{\lceil v}\). From statement (iv) in Theorem~\ref{thm: Nakamoto-implies-longest-chain}, it follows that \(\Chain{h}{t}\) contains \(b_j\).

The above argument is true for all \(t \in [\tau_j^k, \tau_j^{k+1} )\), so it follows that \(h \in \Chain{h}{t}\) for all \(t \in [\tau_j^k,\tau_j^{k+1})\). Since this is true for all \(k \geq k_0\), it follows that \(h \in \Chain{h}{t}\) for all \(t \geq \tau_j^{k_0}\), as desired.
}

%% file: Proofs/Pf-Independence_of_B.tex
\begin{replemma}{lem: B-are-independent} 
    \input{Theorem-Statements/Thm-Independence-of-B}
\end{replemma}
\pf{ 
We prove only the first statement, since the second uses a similar argument. For any \(m \in \{ i,j,j',k\}\), let \(\tau_m\) denote the mining time of \(b_m\). The LHS of the event
\alns{ 
\Bf{j'}{k}{\ss} &\colon \A{b_{j'}}{b_k} \geq \ss \cdot \FS{j'}{k-1} - \FU{j'}{k}
}
depends on the number of adversarial arrivals in \( [\tau_{j'},\tau_{k}]\). Further, the RHS depends on the delays from \(b_{j'}\) to all the honest blocks mined in \( [\tau_{j'},\tau_{k}) \). 

In contrast, the LHS of the event
\alns{ 
\Bb{i}{j}{\ss} &\colon \A{b_i}{b_j} \geq  \ss \cdot \BS{j}{i+1}  - \BU{j}{i}
}
depends on the number of adversarial arrivals in \( [\tau_{i},\tau_{j}]\). Further, the RHS depends on the delays from all the honest blocks mined in \( [\tau_{i},\tau_{j}) \) to \(b_{j}\). 

Since honest and adversarial arrivals are independent Poisson processes, and since the delay associated with any two blocks in \( [\tau_{i},\tau_{j}]\) is independent of the delay associated with any two blocks in \( [\tau_{j'},\tau_{k}]\), and since the two intervals do not overlap, it follows that \( \Bb{i}{j}{\ss}\) and \( \Bf{j'}{k}{\ss}\) are independent.
}

%% file: Proofs/Pf-Probability_Lower_Bound.tex
\begin{reptheorem}{thm: Lower-bound-on-Nakamoto} 
    \input{Theorem-Statements/Thm-Probability-Lower-Bound}
\end{reptheorem}
\pf{ 
Fix \(j\). Recall that the event \( \Njss = \Nb{j}{\ss} \bigcap \Nf{j}{\ss}\), where \(\Nb{j}{\ss} = \bigcap\limits_{i\colon i < j} \Bbc{i}{j}{\ss}\) and \(\Nf{j}{\ss} = \bigcap\limits_{k\colon k > j} \Bfc{j}{k}{\ss}\). From Lemma~\ref{lem: B-are-independent}, it follows that \(\Nb{j}{\ss}\) and \(\Nf{j}{\ss}\) are independent events. Therefore, 
\alns{ 
\P{\Njss} = \P{\Nb{j}{\ss} \bigcap \Nf{j}{\ss}} = \P{\Nb{j}{\ss}} \P{\Nf{j}{\ss}}.
}
Thus, it suffices to show the existence of \( p > 0 \) such that \(\P{\Nb{j}{\ss}} \geq p\) and \( \P{\Nf{j}{\ss}} \geq p\).  

\paragraph{Notation:} If an i.i.d. random process is a sequence of random variables with known distribution, say Geometric or Bernoulli with parameter \(q\), we refer to the \(i\)-th random variable in the sequence as \( \mathsf{Geom}_i(q)\) or \(\mathsf{Be}_i(q)\) respectively.

\vskip 0.2 in
\noindent \textbf{Showing \( \P{ \Nf{j}{\ss}} \geq p > 0\):}  By definition, the event \(\Nf{j}{\ss}\) occurs if
\aln{ 
\A{b_j}{b_k} - \ss\cdot\FS{j}{k-1} < - \FU{j}{k} \ \forall k > j. \label{eq: Ev_2-occurrence}
}
Notice that the LHS are RHS are independent random variables. The random variables inn the LHS are:
\aln{ 
\A{b_j}{b_k} &= \sum_{i=1}^{k-j}  \pbr{ \mathsf{Geom}_{i}\pbr{1-\beta} - 1}, \label{eq: Adv-as-sum} \\
\FS{j}{k-1} &= 1 + \sum_{i=1}^{k-j-1} \mathsf{Be}_{i}\pbr{1-d}, \label{eq: FRSH-as-sum}
}
where equation~\eqref{eq: Adv-as-sum} follows from the fact that there are \( \mathsf{Geom}\pbr{1-\beta} - 1 \) adversarial arrivals between two successive honest arrivals, and the fact that the number of adversarial arrivals in disjoint intervals is independent. Equation~\eqref{eq: FRSH-as-sum} follows from Remark~\ref{rem: FS-Dist}.

Let \(\pbr{\X_i \colon i \geq 1}\) be an i.i.d. random process, with \( \X_i \sim \mathsf{Geom}_i\pbr{1-\beta} -1 - \ss \cdot\mathsf{Be}_i\pbr{1-d}\). Let \(\W_{j} = \sum_{i = 1}^{j} \X_i\) denote the sum of the first \(j\) terms of the process \(\pbr{\X_i \colon i \geq 1} \). Finally, it follows from Remark~\ref{rem: FU-Dist} that \( \pbr{\FU{j}{k} \colon k > j}\) is identical to the i.i.d. random process \( \pbr{\Y_i\colon i > 0}\), with \(\Y_i \sim \mathsf{Geom}_i\pbr{1-d} - 1 \). Using these random variables, equation~\eqref{eq: Ev_2-occurrence} can be equivalently stated as
\aln{ 
\sum_{i=1}^{k-j} \X_i < -\Y_{k-j} \quad \forall k >j. \label{eq: kth-desired-inequality}
}
Let \(\sep\) be a constant to be determined later, such that \( 0 < \sep < \ss\). The inequality in~\eqref{eq: kth-desired-inequality} holds if \( \EA \cap \EB \) occur, where
\alns{ 
\EA &\colon \sum_{j=1}^{i} \X_j < - \sep i \quad \forall i \geq 1 \\
\EB &\colon -\sep i < - \mathsf{Y}_i \quad \forall i \geq 1.
}
Here, \( \EA\) and \( \EB \) are independent events. Therefore, we have \( \P{\Nf{j}{\ss}} \geq \P{\EA \cap \EB} = \P{\EA} \P{\EB}\). It suffices to show the existence of \(p_1,p_2 > 0\) such that \(\P{\EA} \geq p_1 \) and \( \P{\EB} \geq p_2\). First, we bound \( \P{\EB}\).
\alns{ 
\P{\EB} = \P{\bigcap_{i \geq 1} \cbr{\Y_i \leq \sep i}} = \prod_{i=1}^{\infty} \pbr{1-d^{1+\sep i}} =: p_2 > 0,
}
since \( \pbr{ \Y_i\colon i \geq 1} \) is an i.i.d process with \( 1+\Y_i \sim \mathsf{Geom}_i\pbr{1-d}\). Notice that \( p_2 > 0 \) for all \( 0 \leq d < 1\) and for all \( \sep > 0\).

Next, we bound \( \P{\EA} \). Fix some \( \ell \in \mathbb{N} \), and let \(c\) be a positive constant given by \(c = \pbr{\ss - \sep}\ell \). Consider the following two desirable events:
\alns{ 
\G_1 & \colon \W_{\ell} \leq -\sep \ell - c \\
\G_2 & \colon \max_{m \geq 0} \pbr{ \W_{\ell + m} - \W_{\ell} + \sep m } < c
}
It is clear that \(\G_1 \cap \G_2 \implies \EA \), and that \( \G_1 \) and \(\G_2\) are independent events, since Poisson arrivals over disjoint intervals are independent. Therefore, it suffices to find constants \(p_{11}, p_{12} > 0\) such that \(\P{\G_1} \geq p_{11} \) and \( \P{\G_2} \geq p_{12}\).
\aln{ 
\P{\G_1} &= \P{\sum_{i=1}^{\ell}\X_i \leq -\sep \ell - c} \\
&= \P{\sum_{i=1}^{\ell}\pbr{ \mathsf{Geom}_i\pbr{1-\beta} -1 - \ss \cdot \mathsf{Be}_i\pbr{1-d} }\leq -\sep \ell - c} \\
&= \P{ \sum_{i=1}^{\ell} \pbr{ \mathsf{Geom}_i\pbr{1-\beta} - \ss\cdot \mathsf{Be}_i\pbr{1-d}} \leq \pbr{1-\sep} \ell - c} \label{eq: X-pre}
\\
&\geq \P{\bigcap_{i=1}^{\ell} \sbr{ \cbr{ \mathsf{Geom}_i\pbr{1-\beta} = 1} \cap \cbr{\mathsf{Be}_i\pbr{1-d} = 1} }} \label{eq: X}\\
&= p_{11} > 0,
}
for fixed \( \ell\), since \(\sep < \ss\) and since \( c \) is chosen to be sufficiently small. Here, the inequality in~\eqref{eq: X} holds because one way to satisfy the inequality in the event in~\eqref{eq: X-pre} is when both LHS and RHS of the inequality equal \( \pbr{1-\ss} \ell \).

It remains to show \(p_{12} > 0\). Consider the random process \( (\sfZ_i\colon i \geq 1) \), where \( \sfZ_i = \mathsf{X}_i + \sep\). Then, \( \pbr{\sum_{i=1}^{m} \sfZ_{i}\colon m \geq 1}\) and \( \pbr{\W_{\ell+m} - \W_{\ell} + \sep m\colon m \geq 1 } \) follow the same distribution. The Kingman bound~\cite{Kingman_1964} yields,
\alns{ 
\P{\G_2} = \P{\max_{m \geq 0 } \sum_{i=1}^{m} \sfZ_i < c} \geq 1 - e^{-\theta^* c} =: p_{12},
}
where
\alns{ 
\theta^* = \sup \cbr{\theta > 0\colon \mathbb{E}\sbr{e^{\theta \sfZ_1} \leq 1}}.
}
Since \(c > 0\), we see that \( p_{12} > 0 \) if \( \theta^* > 0\). We show that \(\theta^* > 0\) if \(\frac{\beta}{1-\beta} < \ss \pbr{1 - d} \). A simple computation yields
\alns{ 
\E{\sfZ_1} &= \frac{\beta}{1-\beta} - \ss \cdot \pbr{1-d} + \sep, 
}
If \(\frac{\beta}{1-\beta} < \ss \pbr{1 - d}\), then there exists \( 0 < \sep < \ss \) such that \( \E{\sfZ_1} < 0\). Since \(\E{e^{\theta\sfZ_1}}_{\theta = 0} = 1\), and \( \frac{d}{d\theta} \E{e^{\theta\sfZ_1}}_{\theta = 0} = \E{\sfZ_1} < 0\), we know that there exists \(\theta_1 > 0\) such that \(\E{e^{\theta\sfZ_1}}_{\theta = \theta_1} < 1\), which then implies \( \theta^* > \theta_1 > 0\). 

The above argument is summarized as
\aln{ 
\P{\Nf{j}{\ss}} \geq \P{\EA}\P{\EB} \geq \P{\G_1} \cdot \P{\G_2} \cdot \P{\EB} \geq p_{11} \cdot p_{12} \cdot p_{2} =: p > 0. \label{eq: p}
}

\vskip 0.2 in
\noindent \textbf{Showing \( \P{ \Nb{j}{\ss}} \geq p > 0\):} For any \(i,k \geq 0\) such that \(j-i = k-j\), it follows from Remarks~\ref{rem: FS-Dist},~\ref{rem: FU-Dist},~\ref{rem: BS-Dist} and~\ref{rem: BU-Dist} that \(\P{\Bb{i}{j}{\ss}} = \P{\Bf{j}{k}{\ss}}\). By symmetry, we have
\alns{ 
\P{\Nb{j}{\ss}} &= \P{\bigcap_{i = 0}^{j} \Bbc{i}{j}{\ss}}\\
& = \P{\bigcap_{k=j+1}^{2j+1} \Bfc{j}{k}{\ss}} \\
& \geq \P{\bigcap_{k: k > j} \Bfc{j}{k}{\ss}} \\
& = \P{\Nf{j}{\ss}} \\
& \geq p, 
}
where the constant \(p > 0\) is the same as in~\eqref{eq: p}. This concludes the proof.
}

%% file: Proofs/Pf-Probability_of_Catchups.tex
\begin{replemma}{lem: Probability-of-catchup-events}
    \input{Theorem-Statements/Thm-Probability-of-Catchups}
\end{replemma}
\pf{ 
First, we show the existence of \(c > 0\) such that \(\P{ \Bf{j}{k}{\ss} } \leq e^{-c \pbr{ k-j} }\). Recall that the event \( \Bf{j}{k}{\ss} \) is defined as
\alns{ 
\Bf{j}{k}{\ss} \colon \ \ \A{b_j}{b_k} \geq \ss \cdot \FS{j}{k-1} - \FU{j}{k}.
}
From Remarks~\ref{rem: FS-Dist},~\ref{rem: FU-Dist},~\ref{rem: BS-Dist} and~\ref{rem: BU-Dist}, these random variables are characterized as:
\alns{ 
\A{b_j}{b_k} &= \sum_{i=1}^{k-j}  \pbr{ \mathsf{Geom}_{i}\pbr{1-\beta} - 1}, 
\\
\FS{j}{k-1} &= 1 + \sum_{i=1}^{k-j-1} \mathsf{Be}_{i}\pbr{1-d}, 
\\
\FU{j}{k} &= \mathsf{Geom}\pbr{1-d} - 1. 
}
Further, the three random variables \( \A{b_j}{b_k}, \FS{j}{k-1}, \FU{j}{k} \) are mutually independent. Consider the following desirable events associated with them.
\alns{ 
\G_1 &\colon \A{b_j}{b_k} < \frac{\beta}{1-\beta} \pbr{k-j} \pbr{1+\eps} \\
\G_2 &\colon \FS{j}{k-1} > \pbr{1-d}\pbr{k-j}\pbr{1-\delta} \\
\G_3 &\colon \FU{j}{k} \leq \pbr{k-j}\sbr{ \pbr{1-\delta} \ss \cdot \pbr{1-d} - \frac{\beta}{1-\beta}\pbr{1+\eps}} 
}
Clearly, \( \G_1 \cap \G_2 \cap \G_3 \implies \Bfc{j}{k}{\ss} \). Therefore, we have
\aln{ 
\P{\Bfc{j}{k}{\ss}} \geq \P{\G_1 \cap \G_2 \cap \G_3}  = \P{\G_1} \P{\G_2} \P{\G_3}. \label{eq: Bfc-main}
}
It suffices to find bounds for each term separately in the RHS of~\eqref{eq: Bfc-main}.

\paragraph{Bounding \(\G_1\):} Since \( \E{\A{b_j}{b_k}} = \frac{\beta}{1-\beta} \pbr{k-j}\), it follows from the Chernoff bound that
\alns{
\alnd{ 
\P{\G_1^{\mathsf{c}}}&=\P{\A{b_j}{b_k} \geq \frac{\beta}{1-\beta} \pbr{k-j} \pbr{1+\eps}} \leq e^{- c_1 \pbr{k-j}}, \\
\text{where } c_1 &= \frac{\beta}{1-\beta}\pbr{1+\eps}\log\pbr{1+\eps} + \frac{1}{1-\beta}\pbr{1+\beta\eps} \log\pbr{1+\beta\eps} > 0 
}
}

\paragraph{Bounding \( \G_2\):} It follows from the Hoeffding bound that
\alns{ 
\P{\G_2^{\mathsf{c}}} &= \P{ 1 + \sum_{i=1}^{k-j-1} \mathsf{Be}_{i}\pbr{1-d} \leq  \pbr{1-d}\pbr{k-j}\pbr{1-\delta}} \\
& \leq \P{ \sum_{i=1}^{k-j} \mathsf{Be}_{i}\pbr{1-d} \leq  \pbr{1-d}\pbr{k-j}\pbr{1-\delta} } \\
& \leq e^{-c_2\pbr{k-j}}, 
}
where \(c_2 = \frac{\delta^2 \pbr{1-d}}{2} > 0\).

\paragraph{Bounding \( \G_3\):} Consider the quantity \( c_4 = \sbr{ \pbr{1-\delta}\ss \cdot \pbr{1-d} - \frac{\beta}{1-\beta}\pbr{1+\eps}}\). If \( \frac{\beta}{1-\beta} < \ss \pbr{ 1-d} \), then there exist \( \ \eps,\delta > 0\) such that \( c_4 > 0 \). Since \(1+\FRU_j\pbr{b_k}\) follows a geometric distribution, we have
\alns{ 
\P{\G_3^{\mathsf{c}}} &=  \P{ \mathsf{Geom}\pbr{1-d} - 1 >  c_4\pbr{k-j} } \\
&= \P{\mathsf{Geom}\pbr{1-d} > 1 + c_4\pbr{k-j}} \\
&= 1 - d^{c_4\pbr{k-j}} \\
&\geq 1 - e^{c_3\pbr{k-j}},
}
for some \(c_3 > 0\).

\vskip 0.2 in
\noindent Combining these facts together, we revisit~\eqref{eq: Bfc-main}. We have
\aln{ 
\P{\Bfc{j}{k}{\ss}} &\geq \P{\G_1} \P{\G_2} \P{\G_3} \\
&= \pbr{1 - \P{\G_1^{\mathsf{c}}}}\pbr{1 - \P{\G_2^{\mathsf{c}}}}\pbr{1 - \P{\G_3^{\mathsf{c}}}} \\
&\geq \pbr{1 - e^{-c_1(k-j)}} \pbr{1 - e^{-c_2(k-j)}} \pbr{1 - e^{-c_3(k-j)}} \\
&\geq \pbr{1-e^{-c_0(k-j)}}^3 \\
&\geq 1-e^{-c\pbr{k-j}}, \label{eq: Bf-constant}
}
where \(c_0 = \max\cbr{c_1,c_2,c_3} > 0\), and subsequently \(c > 0\). We conclude the existence of \(c>0\) for which \(\P{\Bf{j}{k}{\ss}} \leq e^{-c \pbr{k-j}}\).

It remains to show that \(\P{\Bb{i}{j}{\ss}} \leq e^{-c \pbr{j-i}}\). The proof is very similar, so the details are omitted. Recall that
\alns{ 
\Bb{i}{j}{\ss} \colon \ \ \A{b_i}{b_j} \geq \ss \cdot \BS{j}{i+1} - \BU{j}{i},
}
where the random variables involved may be written as
\alns{ 
\A{b_i}{b_j} &= \sum_{k=1}^{j-i}  \pbr{ \mathsf{Geom}_{k}\pbr{1-\beta} - 1}, 
\\
\BS{j}{i+1} &= 1 + \sum_{k=1}^{j-i-1} \mathsf{Be}_{k}\pbr{1-d}, 
\\
\BU{j}{i} &= \mathsf{Geom}\pbr{1-d} - 1, 
}
Thus, if \(k-j\) = \(j-i\), we see that
\begin{itemize}
    \item \(\A{b_i}{b_j} \) and \( \A{b_j}{b_k}\) follow the same distribution.
    \item \(\BS{j}{i+1} \) and \(\FS{j}{k-1}\) follow the same distribution.
    \item \(\BU{j}{i} \) and \( \FU{j}{k} \) follow the same distribution.
\end{itemize}
Therefore, the same concentration inequalities apply, and we conclude that for the same constant \(c\) in~\eqref{eq: Bf-constant}, we have \(\P{\Bb{i}{j}{\ss}} \leq e^{-c\pbr{j-i}}\).
}

%% file: Proofs/Pf-Sqrt_Decay.tex
\begin{reptheorem}{thm: Exp-sqrt-decay} 
    \input{Theorem-Statements/Thm-Sqrt-Decay}
\end{reptheorem}
\pf{ 
For any \(i\geq 0\), let \(\tau_i\) denote the mining time of block \(b_i\). Partition the interval \( [s,s+t] \) into \( \sqrt{t} \) intervals of length \( \sqrt{t} \) each. Group these sub-intervals into threes, so that there are \(\sqrt{t}/3\) groups of sub-intervals, namely \( I_{1} \), \( I_{2} \), \( \cdots \), \( I_{\sqrt{t}/3} \). Thus, \(I_{\ell} = \sbr{s+3\pbr{\ell - 1}\sqrt{t}, s+3\ell \sqrt{t}} \). Further, let \( S_{\ell} \) represent the middle sub-interval of \( I_{\ell} \), so that \( S_{\ell} = [s+ \pbr{3 \ell - 2} \sqrt{t}, s + \pbr{3 \ell - 1}  \sqrt{t}]\), as shown in Figure~\ref{fig: Intervals-Split}. 

\begin{figure}[t]
    \centering
    \includegraphics[width = 0.99 \textwidth]{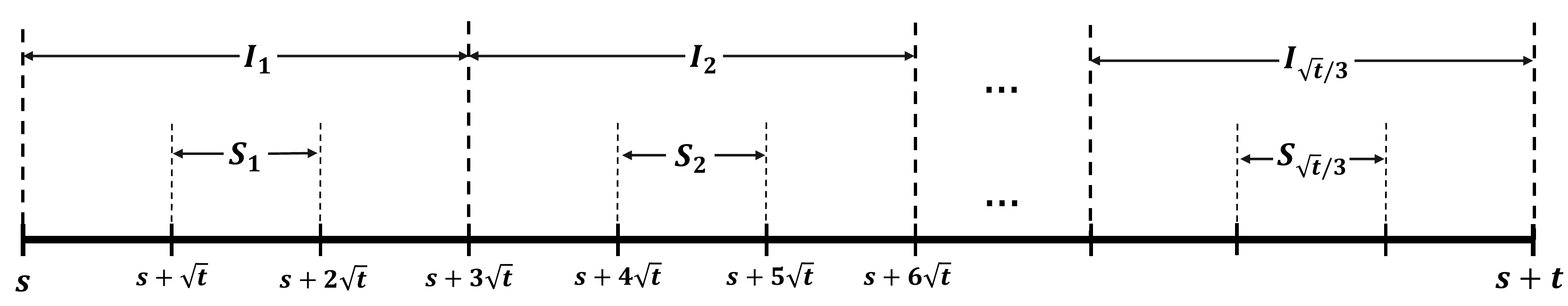}
    \caption{Partitioning \([s,s+t]\) into sub-intervals}
    \label{fig: Intervals-Split}
\end{figure}

Consider the following desirable events.
\alns{ 
\G_1 \colon \bigcap_{j\colon \tau_{j} \in [s+\sqrt{t},s+t-\sqrt{t}]}  \sbr{ \pbr{ \bigcap_{ \substack{ i < j \\ \tau_i < \tau_j - \sqrt{t}  } } \Bbc{i}{j}{\ss} } \bigcap \pbr{ \bigcap_{ \substack{ k > j \\ \tau_k > \tau_j + \sqrt{t} } } \Bfc{j}{k}{\ss}} } \\
\G_2 \colon \bigcup_{\ell = 1}^{\sqrt{t}/3} \mathsf{H}_{\ell}, \text{ where } \mathsf{H}_{\ell} \colon \bigcup_{j\colon \tau_j \in S_{\ell}}\sbr{\pbr{ \bigcap_{ \substack{ i < j \\ \tau_i \geq \tau_j - \sqrt{t} } } \Bbc{i}{j}{\ss} } \bigcap \pbr{ \bigcap_{ \substack{ k > j \\ \tau_k \leq \tau_j + \sqrt{t} } } \Bfc{j}{k}{\ss}}}
}
First, observe that \( \G_1 \cap \G_2 \implies \sbr{B_{s,s+t}^{\pbr{\ss}}}^{\c}\). This is because \( \G_1 \) ensures that no catch-up events (neither \( \Bb{i}{j}{\ss} \), nor \(\Bf{j}{k}{\ss} \)) occur when \(b_i\) and \(b_k\) are separated in time by more than \( \sqrt{t}\) from \(b_j\), for any \(b_j\) mined in \([s+\sqrt{t},s+t-\sqrt{t}]\). This means that the existence of a block \(b_j\) in this interval for which no catch-up event occurs whenever \(i\) and \(k\) are within \(\sqrt{t}\) of \(\tau_j\) is sufficient to ensure that \(b_j\) is a Nakamoto block. This is exactly the event \(\G_2\). Since \( \cup_{\ell} S_{\ell} \subset \sbr{s+\sqrt{t},s+t-\sqrt{t}}\), it follows that \( \G_1 \cap \G_2 \implies \sbr{B_{s,s+t}^{\pbr{\ss}}}^{\mathsf{c}} \). Thus,
\alns{ 
\P{B_{s,s+t}^{\pbr{\ss}}} \leq \P{\Gc_1} + \P{\Gc_2}.
}
Next, we bound the probability of each term in the RHS separately.

\paragraph{Bounding \(\P{\Gc_1} \):} Fix \(\delta > 0\). Consider the following events: 
\alns{ 
\D_1 & \colon \cbr{ \#\cbr{i\colon \tau_i \in [s,s+t]} > 2\lambda_h t }, \\
\D_2 &\colon \cbr{\exists \ \tau_i , \tau_k \in [s,s+t]\colon \pbr{k-i} < \pbr{1-\delta} \lambda_h \sqrt{t}, \tau_k - \tau_i > \sqrt{t}}.
}
By the tail bound for Poisson random variables, we know that \(\P{\D_1} \leq e^{-c_0 t} \) for some \(c_0 > 0\). We now show that \(\P{\D_2} \leq e^{-c_1\sqrt{t}} \) for some \(c_1 > 0\). Let \(T_{i,k} := \tau_k - \tau_i\) be the random variable denoting the time between the \(i\)-th and \(k\)-th mining. Let \(\Mt = \pbr{1-\delta}\lambda_h \sqrt{t} \). Notice that
\alns{ 
\bigcap_{i \in \sbr{s,s+t}} \cbr{ T_{i,i+\Mt} < \sqrt{t}} \implies \D_2^{\mathsf{c}}.
}
For any honest arrival time \(\tau_i\), we have
\alns{ 
\E{T_{i,i+\Mt}} = \E{\sum_{j = i}^{i+\Mt} \pbr{\tau_j - \tau_{j-1}}} = \frac{\Mt}{\lambda_h} = \pbr{1-\delta}\sqrt{t}. 
}
Applying the Chernoff bound, we see that there exists \(c > 0\) such that
\alns{
\P{T_{i, i+\Mt} >  \sqrt{t} } \leq \P{ T_{i, i+\Mt} > \pbr{1-\delta} \sqrt{t}  + \delta\sqrt{t}}  \leq e^{-c\sqrt{t}}.
}
Thus, we get
\alns{ 
\P{\D_2} &= \P{\bigcup_{ i\colon \tau_i \in [s,s+t] } \cbr{T_{i,i + \Mt} > \sqrt{t} }} \\
& \leq \P{ \pbr{ \bigcup_{ i\colon \tau_i \in [s,s+t]} \cbr{T_{i,i+\Mt} > \sqrt{t} } } \cap \Dc_1 } + \P{\D_1} \\
& \leq \pbr{\sum_{i=1}^{2\lambda_h t} \P{T_{i,i+\Mt} > \sqrt{t} } } + e^{-c_0 t}\\
& \leq e^{-c_1 \sqrt{t}},
}
for some \(c_1 > 0\). Thus, \( \P{\D_2} \leq e^{-c_1 \sqrt{t}} \). We may therefore bound \( \P{\Gc_1} \) as
\alns{ 
\P{\Gc_1} &\leq \P{\D_1 \cup \D_2} + \P{\Gc_1 \cap \Dc_1 \cap \Dc_2 } \\
& \leq \P{\D_1} + \P{\D_2} + \P{\Gc_1 \cap \Dc_1 \cap \Dc_2 } \\
& \leq e^{-c_0 t} + e^{-c_1 \sqrt{t}} + \sum_{j=1}^{2 \lambda_h t} \sbr{ \pbr{\sum_{i = 0}^{j-\Mt} \P{\Bb{i}{j}{\ss}}} + \pbr{\sum_{k= j + \Mt }^{\infty} \P{\Bf{j}{k}{\ss}}}} \\
&\leq e^{-c_0 t} + e^{-c_1 \sqrt{t}} + \sum_{j=1}^{2 \lambda_h t} \sbr{ \pbr{\sum_{i = 0}^{j-\Mt} e^{-c\pbr{j-i}}} + \pbr{\sum_{k= j + \Mt }^{\infty} e^{-c\pbr{k-j}}}} \\
&\leq e^{-c_0 t} + e^{-c_1 \sqrt{t}} + \sum_{j=1}^{2 \lambda_h t} \pbr{2\sum_{m = \Mt}^{\infty} e^{-cm} } \\
& = e^{-c_0 t} + e^{-c_1 \sqrt{t}} + \frac{4\lambda_h t}{1-e^{-c}} e^{-c \Mt} \\
& = e^{-c_0 t} + e^{-c_1 \sqrt{t}} + \frac{4\lambda_h t}{1-e^{-c}} e^{-c \pbr{1-\delta}\lambda_h\sqrt{t}} \\
& \leq e^{-c_3 \sqrt{t}},
}
for some \(c_3 > 0\).

\paragraph{Bounding \( \P{\Gc_2} \):} We have \(\Gc_2 = \bigcap_{\ell = 1}^{\sqrt{t}/3} \Hc_{\ell} \). Notice that \(\Hc_{\ell} \) are mutually independent for distinct \(\ell\) by Lemma~\ref{lem: B-are-independent}. Recall that 
\alns{ 
\sfH_{\ell} = \bigcup_{j \colon \tau_j \in S_{\ell}} \mathsf{R}_{j}^{\ell}, \text{ where } \mathsf{R}_{j}^{\ell} \colon \pbr{ \bigcap_{ \substack{ i < j \\ \tau_i \geq \tau_j - \sqrt{t} } } \Bbc{i}{j}{\ss} } \bigcap \pbr{ \bigcap_{ \substack{ k > j \\ \tau_k \leq \tau_j + \sqrt{t} } } \Bfc{j}{k}{\ss}}.
}
\noindent Let \(M_{\ell}\) be the number of honest blocks mined in \(S_{\ell}\), and \( N_{\ell}\) be the number of \(\ss\)-Nakamoto blocks mined in \(S_{\ell}\). Since \(\sfH_{\ell}\) is contained in the event \(M_{\ell } \geq 1\), we have for each \(\ell \in \cbr{1,2,\cdots,\sqrt{t}/3}\): 
\aln{ 
\P{\sfH_{\ell}} &=  \P{ \bigcup_{j \colon \tau_j \in S_{\ell}} \mathsf{R}_{j}^{\ell} } \nonumber\\
& = \P{ \pbr{\bigcup_{j \colon \tau_j \in S_{\ell}} \mathsf{R}_{j}^{\ell} } \bigcap \cbr{M_{\ell} \geq 1} } \nonumber\\
& \geq \P{ \bigcup_{j: \tau_j \in S_{\ell}} \Njss} \nonumber \\
& \geq \frac{p_0^2}{2},\label{eq: T2}
}
where \(p_0\) is the lower bound on the probability that \(b_j\) is an \(\ss\)-Nakamoto block, obtained in Theorem~\ref{thm: Lower-bound-on-Nakamoto}. The inequality in~\eqref{eq: T2} deserves some elaboration: Since \(M_{\ell}\) is a non-negative integer valued random variable, we have from the second moment method that
\alns{ 
\P{N_{\ell} \geq 1} = \P{N_{\ell} > 0} \geq \frac{\pbr{\E{N_{\ell}}}^2}{\E{N_{\ell}^2}} \geq \frac{  \pbr{ p_0\lambda_h \sqrt{t}}^2 t}{\lambda_h \sqrt{t} + \pbr{\lambda_h \sqrt{t}}^2} \geq \frac{\lambda_h^2 p_0^2 t}{2 \lambda_h^2 t} \geq \frac{p_0^2}{2},
}
for sufficiently large \(t\). Here, we used the fact that \( \E{N_{\ell}} = p_0 \lambda_h  \sqrt{t}\), and \( \E{N_{\ell}^2} \leq \E{M_{\ell}^2} = \lambda_h \sqrt{t} + \pbr{\lambda_h \sqrt{t}}^2\).

Thus, \( \P{\Hc_{\ell}} \leq 1 - \frac{p_0^2}{2} < 1\), which yields
\alns{ 
\P{ \Gc_2} = \P{\bigcap_{\ell=1}^{\sqrt{t}/3} \Hc_{\ell}}  = \prod_{\ell = 1}^{\sqrt{t}/3} \P{\Hc_{\ell}} \leq \pbr{1-\frac{p_0^2}{2}}^{\sqrt{t}/3} \leq e^{-c_4 \sqrt{t}}, 
}
for some \(c_4 > 0\), since the \( \sfH_{\ell}\)'s are mutually independent events. Therefore, we conclude that
\alns{ 
\P{B_{s,s+t}^{\pbr{\ss}}} \leq \P{\Gc_1} + \P{\Gc_2} \leq e^{-c_3\sqrt{t}} + e^{-c_4 \sqrt{t}} \leq e^{-c_0 \sqrt{t}},
}
for some \(c_0 > 0\), as desired.
}

%% file: Proofs/Pf-Bootstrap.tex
\begin{reptheorem}{thm: Bootstrap} 
    \input{Theorem-Statements/Thm-Bootstrap}
\end{reptheorem}
\pf{ 
Fix \( m > 1\). Consider the following statement for \(m\):
\alns{ 
\S[m] \  \colon \forall \ \theta \geq m, \ \exists \ a_{\theta} > 0, A_{\theta} > 0 \text{ such that } \P{B_{s,s+t}^{\pbr{\ss}}} \leq A_{\theta} \exp \pbr{-a_{\theta} t^{1/\theta}}
}
In Theorem~\ref{thm: Exp-sqrt-decay}, we proved that \( \S[2]\) is true. Next, we show the following:
\alns{ 
\S[m] \implies \S\sbr{\frac{2m-1}{m}}.
}
Assume \( \S[m]\) is true. For any \(i\geq 0\), let \(\tau_i\) denote the mining time of block \(b_i\). Partition the interval \( [s,s+t]\) into \( t^{\frac{m-1}{2m-1}} \) intervals of length \(t^{\frac{m}{2m-1}}\) each. Group these sub-intervals into threes, so that there are \(  \frac{t^{\frac{m-1}{2m-1}}}{3}\) groups of sub-intervals, namely \(I_{1}, I_{2}, \cdots, I_{t^{\frac{m-1}{2m-1}}/3} \). Thus,
\alns{ 
I_{\ell} &\colon = \sbr{s+ 3 \pbr{\ell - 1} t^{\frac{m}{2m-1}}, s + 3\ell t^{\frac{m}{2m-1}} }, \\
S_{\ell} &\colon = \sbr{s+  \pbr{3\ell - 2} t^{\frac{m}{2m-1}}, s + \pbr{3\ell - 1} t^{\frac{m}{2m-1}} }.
}
Consider the following desirable events.
\alns{ 
\G_1 \colon \bigcap_{j\colon \tau_{j} \in \sbr{s+t^{\frac{m}{2m-1}},s+t-t^{\frac{m}{2m-1}}}}  \sbr{ \pbr{ \bigcap_{ \substack{ i < j \\ \tau_i < \tau_j - t^{\frac{m}{2m-1}}  } } \Bbc{i}{j}{\ss} } \bigcap \pbr{ \bigcap_{ \substack{ k > j \\ \tau_k > \tau_j + t^{\frac{m}{2m-1}} } } \Bfc{j}{k}{\ss}} } \\
\G_2 \colon \bigcup_{\ell = 1}^{t^{\frac{m-1}{2m-1}}/3} \mathsf{H}_{\ell}, \text{ where } \mathsf{H}_{\ell} \colon \bigcup_{j\colon \tau_j \in S_{\ell}}\sbr{\pbr{ \bigcap_{ \substack{ i < j \\ \tau_i \geq \tau_j - t^{\frac{m}{2m-1}} } } \Bbc{i}{j}{\ss} } \bigcap \pbr{ \bigcap_{ \substack{ k > j \\ \tau_k \leq \tau_j + t^{\frac{m}{2m-1}} } } \Bfc{j}{k}{\ss}}}
}
Observe that \( \G_1 \cap \G_2 \implies \sbr{B_{s,s+t}^{\pbr{\ss}}}^\c\).

\paragraph{Bounding \( \P{\Gc_1}\):} We show that \( \P{ \Gc_1} \leq e^{-c_3 t^{\frac{m}{2m-1}}}\) for some \(c_3 > 0\). This is done following the steps in Theorem~\ref{thm: Exp-sqrt-decay}. Fix \( \delta > 0\) and consider the following events:
\alns{ 
\D_1 & \colon \cbr{ \#\cbr{i\colon \tau_i \in [s,s+t]} > 2\lambda_h t }, \\
\D_2 &\colon \cbr{\exists \ \tau_i, \tau_k \in [s,s+t]\colon \pbr{k-i} < \pbr{1-\delta} \lambda_h t^{\frac{m}{2m-1}}, \tau_k - \tau_i > t^{\frac{m}{2m-1}}}.
}
By the tail bound for Poisson random variables, we know that \(\P{\D_1} \leq e^{-c_0 t} \) for some \(c_0 > 0\). We now show that \(\P{\D_2} \leq e^{-c_1 t^{\frac{m}{2m-1}}} \) for some \(c_1 > 0\). Let \(T_{i,k} := \tau_k - \tau_i\) be the random variable denoting the time between the \(i\)-th and \(k\)-th mining. Let \(\Mt = \pbr{1-\delta}\lambda_h t^{\frac{m}{2m-1}} \). Notice that
\alns{ 
\bigcap_{i \in \sbr{s,s+t}} \cbr{ T_{i,i+\Mt} < \sqrt{t}} \implies \D_2^{\mathsf{c}}.
}
For any honest arrival time \(\tau_i\), we have
\alns{ 
\E{T_{i,i+\Mt}} = \E{\sum_{j = i}^{i+\Mt} \pbr{\tau_j - \tau_{j-1}}} = \frac{\Mt}{\lambda_h} = \pbr{1-\delta} t^{\frac{m}{2m-1}}.
}
Applying the Chernoff bound, we see that there exists \(c > 0\) such that
\alns{
\P{ \T_{i, i+\Mt} >  t^{\frac{m}{2m-1}} }  &\leq \P{ T_{i, i+\Mt} > \pbr{1-\delta} t^{\frac{m}{2m-1}}  + \delta t^{\frac{m}{2m-1}}}  \leq e^{-c t^{\frac{m}{2m-1}}}.
}
Thus, we get
\alns{ 
\P{\D_2} &= \P{\bigcup_{ i\colon \tau_i \in [s,s+t] } \cbr{T_{i,i + \Mt} > t^{\frac{m}{2m-1}} }} \\
& \leq \P{ \pbr{ \bigcup_{ i\colon \tau_i \in [s,s+t]} \cbr{T_{i,i+\Mt} > t^{\frac{m}{2m-1}} } } \cap \Dc_1 } + \P{\D_1} \\
& \leq \pbr{\sum_{i=1}^{2\lambda_h t} \P{T_{i,i+\Mt} > t^{\frac{m}{2m-1}} } } + e^{-c_0 t}\\
& \leq e^{-c_1 t^{\frac{m}{2m-1}}},
}
for some \(c_1 > 0\). Thus, \( \P{\D_2} \leq e^{-c_1 t^{\frac{m}{2m-1}}} \). We may therefore bound \( \P{\Gc_1} \) as
\alns{ 
\P{\Gc_1} &\leq \P{\D_1 \cup \D_2} + \P{\Gc_1 \cap \Dc_1 \cap \Dc_2 } \\
& \leq \P{\D_1} + \P{\D_2} + \P{\Gc_1 \cap \Dc_1 \cap \Dc_2 } \\
& \leq e^{-c_0 t} + e^{-c_1 t^{\frac{m}{2m-1}}} + \sum_{j=1}^{2 \lambda_h t} \sbr{ \pbr{\sum_{i = 0}^{j-\Mt} \P{\Bb{i}{j}{\ss}}} + \pbr{\sum_{k= j + \Mt }^{\infty} \P{\Bf{j}{k}{\ss}}}} \\
&\leq e^{-c_0 t} + e^{-c_1 t^{\frac{m}{2m-1}}} + \sum_{j=1}^{2 \lambda_h t} \sbr{ \pbr{\sum_{i = 0}^{j-\Mt} e^{-c\pbr{j-i}}} + \pbr{\sum_{k= j + \Mt }^{\infty} e^{-c\pbr{k-j}}}} \\
&\leq e^{-c_0 t} + e^{-c_1 t^{\frac{m}{2m-1}}} + \sum_{j=1}^{2 \lambda_h t} \pbr{2\sum_{m = \Mt}^{\infty} e^{-cm} } \\
& = e^{-c_0 t} + e^{-c_1 t^{\frac{m}{2m-1}}} + \frac{4\lambda_h t}{1-e^{-c}} e^{-c \Mt} \\
& = e^{-c_0 t} + e^{-c_1 t^{\frac{m}{2m-1}}} + \frac{4\lambda_h t}{1-e^{-c}} e^{-c \pbr{1-\delta}\lambda_h t^{\frac{m}{2m-1}}} \\
& \leq e^{-c_3 t^{\frac{m}{2m-1}}},
}
for some \(c_3 > 0\).

\paragraph{Bounding \( \P{\Gc_2} \):} We have \(\Gc_2 = \bigcap_{\ell = 1}^{t^{\frac{m-1}{2m-1}}/3} \Hc_{\ell} \). Notice that \(\Hc_{\ell} \) are mutually independent for distinct \(\ell\) by Lemma~\ref{lem: B-are-independent}. Since \( \S[m]\) is true, we have
\alns{ 
\P{\sfH_{\ell}^\c} \leq \P{B_{s+\pbr{3\ell-2} t^{\frac{m}{2m-1}}, s+ \pbr{3\ell - 1}t^{\frac{m}{2m-1}} }^{\pbr{\ss}} } \leq A_{m} \exp \pbr{- a_{m}  \pbr{t^{\frac{m}{2m-1}}}^{\frac{1}{m}} } \leq A_{m} \exp \pbr{- a_{m} t^{1/(2m-1)}}
}
Therefore, it follows that
\alns{ 
\P{ \Gc_2} = \P{\bigcap_{\ell=1}^{t^{ \frac{m-1}{2m-1}}/3} \Hc_{\ell}}  = \prod_{\ell = 1}^{t^{ \frac{m-1}{2m-1}}/3} \P{\Hc_{\ell}} \leq \pbr{A_m \exp \pbr{-a_m t^{1/(2m-1)}}}^{t^{ \frac{m-1}{2m-1}}/3} \leq e^{-c_4 t^{\frac{m}{2m-1}}}, 
}
for some \(c_4 > 0\), since the \( \sfH_{\ell}\)'s are mutually independent events. Therefore, we have
\alns{ 
\P{B_{s,s+t}^{\pbr{\ss}}} \leq \P{\Gc_1} + \P{\Gc_2} \leq e^{-c_3t^{\frac{m}{2m-1}}} + e^{-c_4 t^{\frac{m}{2m-1}}} \leq A_{m'} e^{- a_{m'} t^{1/m'}},
}
where \( m' = \frac{m}{2m-1} \), and \(a_{m'}, A_{m'} > 0\). In turn, this implies that \( \S\sbr{\frac{2m-1}{m}} \) is true. 

Finally, consider the recursion given by \( m_{k+1} = \frac{2m_k - 1}{m_k}\), and the initial condition \(m_1 = 2 \). We have proved that \( \S[m_1]\) is true and that \( \S\sbr{m_{k+1}} \) is true whenever \( \S[m_k]\) is true. By induction, it follows that \( \S[m_k]\) is true for all \( k \in \N \). Since \(m_k = \frac{k+1}{k} \) and \(\lim_{k \to \infty} m_k = 1\), we conclude that for every \(\eps > 0\), there exist positive constants \(A\), \(a\) such that \(\P{B_{s,s+t}^{\pbr{\ss}}} \leq A \exp \pbr{-a t^{1-\eps}}\). This concludes the proof.
}

%% file: Proofs/Pf-Probabilistic-Observer.tex
\begin{reptheorem}{thm: Probability-Observer-Criterion}
    \input{Theorem-Statements/Thm-Probabilistic-Observer}
\end{reptheorem}

\pf{ 
By the union bound, we have 
\aln{
\P{\bigcup_{k = k'}^{\infty} \cbr{ \Un{h}{J}{k} > \pbr{\frac{1-\ss}{2}} k}} &\leq \sum_{k = k'}^{\infty} \P{ \Un{h}{J}{k} > \pbr{\frac{1-\ss}{2}}k}. \label{eq: Unheard-Ineq}
}
From remark~\ref{rem: Unheard-Dist}, it follows that \( \Un{h}{J}{k} \leq \mathsf{Geom}\pbr{ 1-d}\). Thus, for all \(k \geq 1\), we have
\alns{ 
\P{ \Un{h}{J}{k} > \pbr{\frac{1-\ss}{2}} k} & \leq \P{ \mathsf{Geom}\pbr{1-d} > \pbr{\frac{1-\ss}{2}} k } \\
& \leq C_0 \cdot e^{-c k},
}
for some positive constants \(C_0\) and \(c\). Combining this with \eqref{eq: Unheard-Ineq} yields
\alns{ 
\P{\bigcup_{k = k'}^{\infty} \cbr{ \Un{h}{J}{k} > \pbr{\frac{1-\ss}{2}}k  }} \leq \sum_{k = k'}^{\infty} C_0\cdot e^{-ck} = C \cdot e^{-c k'},
}
where \(C = \frac{C_0}{1-e^{-c}}\) is a positive constant. This concludes the proof.
}

%% file: Proofs/Pf-Main-Result.tex
Before we prove our main result, we recall a useful lemma about Poisson random variables.

\begin{lemma}\label{lem: Poisson}
Let \(X\) be a Poisson random variable with mean \(\mu\). Then
\begin{enumerate}[label = (\roman*)]
    \item \( \P{X \geq 2 \mu} \leq e^{-\frac{4}{3}\mu}\).
    \item \( \P{X \leq \frac{1}{2}\mu} \leq e^{-\frac{1}{8}\mu}\).
\end{enumerate}
\end{lemma}
\pf{
The proof follows from Theorem 4.5 in~\cite{Mitzenmacher_2017}.
}
We are now ready to state and prove our main result.

\begin{reptheorem}{thm: Main}
    \input{Theorem-Statements/Thm-Main-Result}
\end{reptheorem}
\pf{ 
Let \( \ss\) be such that \( \frac{\beta}{1-\beta} < \ss \cdot \pbr{1-d}\). 
Let \(k_0 =  \ceil{\frac{2\ss}{1-\ss}}\), and fix an honest subset of users \(\mathcal{H}\). The idea of the proof is as follows: if an \(\ss\)-Nakamoto block \(b_J\) is mined in the time interval \((s,s+t_1)\), and \(k_0\) number of \(J\dash\FRSH\) blocks are mined before time \(s + t_1 + t_2\), and if all users \(h \in \mathcal{H}\) satisfy the \(\pbr{h,J,\ss,k_0}\)-user-unheard-criterion, then Theorems \ref{thm: Nakamoto-implies-longest-chain} and \ref{thm: Obs-and-Nak} together imply that the \(\ss\)-Nakamoto block \(b_J\) is included in \( \mathcal{C}_{h}(t) \) for all \(t \geq s + t_1 + t_2\). Since \(\mathsf{tx}\) must be included in either \(b_J\) or its ancestors, \(\mathsf{tx}\) satisfies \( \pbr{t_1+t_2,\mathcal{H}}\)-security. 

Let \(b_J\) be the first \(\ss\)-Nakamoto block mined after time \(s\). Let \(T_J = \tau_J - s\) denote the time between \(s\) and the mining time of the first \(\ss\)-Nakamoto block. For the \(r\)-th honest block \(b_r\) and any time \(t\), let \(N_r(t)\) denote the number of \(r\dash\FRSH\) blocks mined until time \(t\). Consider the following events:
\alns{ 
\Ev_1 &\colon T_J > t_1 \\
\Ev_2^{h} &\colon \text{The } \pbr{h,J,\ss,k_0}\text{-user-unheard-criterion is violated} \\
\Ev_3 &\colon \bigcup_{r: \tau_r \in [s, s+t_1]} \cbr{ N_{r}\pbr{s+t_1+t_2} < k_0}.
}
By the first paragraph of the proof, the union bound gives
\aln{ 
\P{\mathsf{tx} \text{ violates \( \pbr{t_1 + t_2, \mathcal{H}}\)-security}} &\leq \P{\Ev_1} + \sum_{h \in \mathcal{H}}\P{\Ev_2^h} + \P{\Ev_3}. \label{eq: Transition}
}

From Theorem \ref{thm: Bootstrap}, we have that for any \(\eps > 0\), there exist positive constants \(A', a'\) such that \(\P{\Ev_1} \leq A' \pbr{\exp{-a' t_1^{1-\eps}}}\). 

From Theorem \ref{thm: Probability-Observer-Criterion}, we have that there exist positive constants \(C\), \(c'\) such that \( \P{\Ev_2^h} \leq C\exp\pbr{-c' k_0}\) for all \(h \in \mathcal{H}\).

It remains to bound \(\P{\Ev_3}\). Let \(\lambda_h\) denote the aggregate mining rate of the honest users. Let \( t_2 > \frac{2 k_0}{\pbr{1-d}\lambda_h}\). Let \(M\) be the number of honest miners in \([s,s+t_1]\), so that \(M\) has the Poisson distribution with mean \( \lambda_h \pbr{1-d}t_1\). Therefore, we have from Lemma~\ref{lem: Poisson} that
\alns{ 
\P{M > 2\lambda_h t_1} \leq \exp\pbr{\frac{-4\lambda_h t_1}{3}}.
}

For \(r \geq 1\), consider the \(r\)-th honest miner \(b_r\) after time \(s\) and consider the \(r\dash\FRSH\) sequence. Let \(U_r\) denote the number of \(r\dash\FRSH\) blocks mined in \( \sbr{\tau_r, \tau_r + t_2}\). Then, \(U_r\) has the Poisson probability distribution with mean \( \lambda_h \pbr{1-d} t_2\). Applying Lemma~\ref{lem: Poisson}, we get
\alns{ 
\P{U_r \leq k_0} \leq \P{U_r \leq \frac{1}{2}\lambda_h \pbr{1-d}t_2} \leq \exp\pbr{\frac{-\lambda_h \pbr{1-d}t_2}{8} }.
}
Therefore, we have
\alns{ 
\P{\Ev_3} &= \P{\Ev_3 \cap \cbr{M > 2 \lambda_h t_1}} + \P{\Ev_3 \cap \cbr{M < 2 \lambda_h t_1}} \\
& \leq \P{N > 2\lambda_h t_1} + \sum_{r=1}^{2 \lambda_h t_1}\P{U_r \leq k_0} \\
& \leq \exp\pbr{\frac{-4\lambda_ht_1}{3}} + 2 \lambda_h t_1 \cdot \exp\pbr{ \frac{-\lambda_h \pbr{1-d}t_2}{8}}.
}
Therefore, combining~\eqref{eq: Transition} and the bounds on \( \P{\Ev_1} \), \(\P{\Ev_2^h}\), and \( \P{\Ev_3}\) yields
\alns{ 
\P{\mathsf{tx} \text{ violates \( \pbr{t_1 + t_2, \mathcal{H}}\)-security}} &\leq \P{\Ev_1} + \sum_{h \in \mathcal{H}} \P{\Ev_2^h } + \P{\Ev_3} \\
&\leq A'\exp\pbr{-a' t_1^{1- \eps}} + \sum_{h \in \mathcal{H}} C\exp\pbr{-c t_2} +  \exp\pbr{\frac{-4 \lambda_h t_1}{3}}  \nonumber \\
& \qquad \qquad \qquad \qquad \qquad \qquad + 2 \lambda_h t_1 \exp\pbr{ \frac{-\lambda_h \pbr{1\!-\!d} t_2}{8}} \\
& \leq A \exp\pbr{-a t_1^{1- \eps}} + B \pbr{|\mathcal{H}| + t_1} \exp\pbr{-bt_2}
}
for some positive constants \(A\), \(a\), \( B\), and \(b\). The following lemma therefore completes the proof
of Theorem \ref{thm: Main}.

\begin{lemma}
Let  \(\eps > 0\).  Suppose there exist positive constants \(A\), \(a\), \(B\), \(b\) such that for all \(t_1 > 0\), \(t_2 > 0\) and for any honest transaction \(\mathsf{tx}\) and any finite set of honest users \( \mathcal{H}\):
\begin{align}   \label{eq:given}
\P{\mathsf{tx} \text{ violates \( \pbr{t_1 + t_2, \mathcal{H}}\)-security}} \leq  A \exp\pbr{-a t_1^{1-\eps}} + B \pbr{|\mathcal{H}|+t_1} \exp \pbr{- b t_2}.
\end{align}
Then there exist positive constants $a'''$ and $b'''$ such that for all $\tau \geq 0$ and for any honest transaction \(\mathsf{tx}\) and any finite set of honest users \( \mathcal{H}\):
\begin{align}  \label{eq:clean_bound}
\P{\mathsf{tx} \text{ violates \( \pbr{\tau, \mathcal{H}}\)-security}} \leq  \exp\pbr{-a''' \tau^{1-\eps}} + |\mathcal{H}| \exp \pbr{- b''' \tau}.
\end{align}
\end{lemma}

\begin{proof}
The lefthand side of \eqref{eq:given} is zero if $\mathcal{H}=\emptyset$ so assume without loss of generality that   $|\mathcal{H}| \geq 1.$
Given $\tau \geq 0$,  let $t_1 = t_2 = \tau/2.$   Then \eqref{eq:given} yields
\begin{align}   \label{eq:step1}
\P{\mathsf{tx} \text{ violates \( \pbr{\tau, \mathcal{H}}\)-security}} \leq  A \exp\pbr{-(a/2^{1-\epsilon}) \tau^{1-\eps}} + 2B \pbr{|\mathcal{H}|+\tau} \exp \pbr{- (b/2)\tau}.
\end{align}
Let $a'$ and $b'$ be positive constants such that $a'< a/2^{1-\epsilon}$ and $b' < b/2.$   Let $\bar{\tau}$ be so  large that
\\ $A \exp\pbr{-[(a/2^{1-\epsilon})-a'] \tau^{1-\eps}} \leq 1$ and $2B \exp \pbr{- [(b/2)-b']\tau}  \leq 1$ for all $\tau \geq \bar{\tau}.$  Then for $\tau \geq \bar{\tau}$
\begin{align}   \label{eq:step2}
\P{\mathsf{tx} \text{ violates \( \pbr{\tau, \mathcal{H}}\)-security}} \leq  \exp\pbr{-a' \tau^{1-\eps}} +  \pbr{|\mathcal{H}|+\tau} \exp \pbr{- b' \tau} 
\end{align}
Let $b''$ be a positive constant with $b'' < b.$   Then, using the assumption  $ |\mathcal{H}| \geq 1,$
\begin{align}
    \pbr{|\mathcal{H}|+\tau} \exp \pbr{- b' \tau}
    & =  |\mathcal{H}| \exp \pbr{- b'' \tau}
     + \tau \exp \pbr{- b' \tau}
     -  |\mathcal{H}|  ( \exp  \pbr{- b'' \tau}  - \exp  \pbr{- b' \tau}  )  \nonumber \\
  &  \leq |\mathcal{H}| \exp \pbr{- b'' \tau}
     + \tau \exp \pbr{- b' \tau}  
     -   ( \exp  \pbr{- b'' \tau}  - \exp  \pbr{- b' \tau}  ) \nonumber  \\
 &    = |\mathcal{H}| \exp \pbr{- b'' \tau}
     -   ( \exp  \pbr{- b'' \tau}  - (1+\tau) \exp  \pbr{- b' \tau}  )  \nonumber \\
& \leq  |\mathcal{H}| \exp \pbr{- b'' \tau} ~~~\mbox{for all $\tau$ sufficiently large}  \label{eq:step3}
\end{align} 
Combining \eqref{eq:step2} and \eqref{eq:step3} implies that there exists $\bar{\tau}'$ such that
\begin{align}   \label{eq:step4}
\P{\mathsf{tx} \text{ violates \( \pbr{\tau, \mathcal{H}}\)-security}} \leq  \exp\pbr{-a' \tau^{1-\eps}} + 
\tau \exp \pbr{- b'' \tau} ~~ \mbox{  for $\tau \geq \bar{\tau}'$}
\end{align}
Select positive constants $a'''$ and $b'''$ such that $a''' < a'$ and $b''' <b''$ and
\begin{align}  \label{eq:step5}
1 \leq  \exp\pbr{-a''' \tau^{1-\eps}} +  \tau \exp \pbr{- b''' \tau} ~~\mbox{ for $0\leq \tau \leq \bar{\tau}'$}
\end{align}
Combining \eqref{eq:step4} and \eqref{eq:step5} yields \eqref{eq:clean_bound} for all $\tau \geq 0.$
\end{proof}
This concludes the proof of Theorem~\ref{thm: Main}.
}